\def\year{2019}\relax
\documentclass[letterpaper]{article} %DO NOT CHANGE THIS
\usepackage{aaai19}  %Required
\usepackage{times}  %Required
\usepackage{helvet}  %Required
\usepackage{courier}  %Required
\usepackage{url}  %Required
\usepackage{graphicx}  %Required
\usepackage{amsmath,amsfonts,amsthm,amssymb}
\usepackage[font=footnotesize]{caption}
\usepackage[caption=false,font=footnotesize,hangindent=10pt]{subfig}
\usepackage[linesnumbered,vlined,boxed,ruled]{algorithm2e}
\usepackage{enumitem}
\usepackage{color,xparse}
\usepackage{tikz}
\usetikzlibrary{shapes.geometric,shapes.misc,positioning,backgrounds,fit,decorations.pathreplacing}

\graphicspath{{./}{figure/}}

\SetAlgoSkip{}
\SetVlineSkip{1pt}

\SetCommentSty{myAlgComment}
\SetKwProg{Fn}{Function}{:}{}

\SetAlFnt{\small}
\SetAlCapFnt{\small}
\SetAlCapNameFnt{\small}

\SetKwFunction{Process}{Process}
\SetKwFunction{ReduceRedundancy}{ReduceRedundancy}
\SetKwProg{myproc}{Procedure}{}{}

\newtheorem{definition}{Definition}

\newtheorem{theorem}{Theorem}
\newtheorem{lemma}{Lemma}

\newcommand{\header}[1]{\smallskip\noindent\textbf{#1}}
\newcommand{\bullethdr}[1]{\noindent\textbullet\,\textbf{#1}}
\newcommand{\remark}{\smallskip\noindent\emph{Remark}}

\newcommand{\CS}{\mathcal{S}}
\newcommand{\CA}{\mathcal{A}}
\newcommand{\bx}{\mathbf{x}}
\newcommand{\OPT}{\mathrm{OPT}}

\title{Submodular Optimization Over Streams with Inhomogeneous Decays}

\author{{\small Junzhou Zhao\textsuperscript{1},
Shuo Shang\textsuperscript{2},
Pinghui Wang\textsuperscript{3},
John C.S. Lui\textsuperscript{4},
and Xiangliang Zhang\textsuperscript{1}\thanks{Shuo Shang and Xiangliang Zhang are the corresponding authors.}}\\
{\small \textsuperscript{1}King Abdullah University of Science and Technology, KSA}\\
{\small \textsuperscript{2}Inception Institute of Artificial Intelligence, UAE}\\
{\small \textsuperscript{3}Xi'an Jiaotong University, China}\\
{\small \textsuperscript{4}The Chinese University of Hong Kong, Hong Kong}\\
{\small \{junzhou.zhao, xiangliang.zhang\}@kaust.edu.sa,
jedi.shang@gmail.com, phwang@mail.xjtu.edu.cn, cslui@cse.cuhk.edu.hk}}

\begin{document}
\maketitle

\begin{abstract}
\begin{quote}
Cardinality constrained submodular function maximization, which aims to select a
subset of size at most $k$ to maximize a monotone submodular utility function, is
the key in many data mining and machine learning applications such as data
summarization and maximum coverage problems.
When data is given as a stream, {\em streaming submodular optimization (SSO)}
techniques are desired.
Existing SSO techniques can only apply to insertion-only streams where each
element has an infinite lifespan, and sliding-window streams where each element
has a same lifespan (i.e., window size).
However, elements in some data streams may have arbitrary different lifespans, and
this requires addressing SSO over streams with {\em inhomogeneous-decays}
(SSO-ID).
This work formulates the SSO-ID problem and presents three algorithms:
\textsc{BasicStreaming} is a basic streaming algorithm that achieves an
$(1/2-\epsilon)$ approximation factor; \textsc{HistApprox} improves the efficiency
significantly and achieves an $(1/3-\epsilon)$ approximation factor;
\textsc{HistStreaming} is a streaming version of \textsc{HistApprox} and uses
heuristics to further improve the efficiency.
Experiments conducted on real data demonstrate that \textsc{HistStreaming} can
find high quality solutions and is up to two orders of magnitude faster than the
naive \textsc{Greedy} algorithm.

\end{quote}
\end{abstract}

\section{Introduction}
\label{sec:introduction}

Selecting a subset of data to maximize some utility function under a cardinality
constraint is a fundamental problem facing many data mining and machine learning
applications.
In myriad scenarios, ranging from data summarization~\cite{Mitrovic2018}, to
search results diversification~\cite{Agrawal2009}, to feature
selection~\cite{Brown2012}, to coverage maximization~\cite{Cormode2010}, utility
functions commonly satisfy {\em submodularity}~\cite{Nemhauser1978}, which
captures the {\em diminishing returns} property.
It is therefore not surprising that submodular optimization has attracted a lot of
interests in recent years~\cite{Krause2014}.

If data is given in advance, the \textsc{Greedy} algorithm can be applied to solve
submodular optimization in a {\em batch} mode.
However, today's data could be generated continuously with no ending, and in some
cases, data is produced so rapidly that it cannot even be stored in computer main
memory, e.g., Twitter generates more than $8,000$ tweets every
second~\cite{tweets_speed}.
Thus, it is crucial to design {\em streaming algorithms} where at any point of
time the algorithm has access only to a small fraction of data.
To this end, {\em streaming submodular optimization (SSO)} techniques have been
developed for {\em insertion-only streams} where a subset is selected from all
historical data~\cite{Badanidiyuru2014a}, and {\em sliding-window streams} where a
subset is selected from the most recent data only~\cite{Epasto2017}.

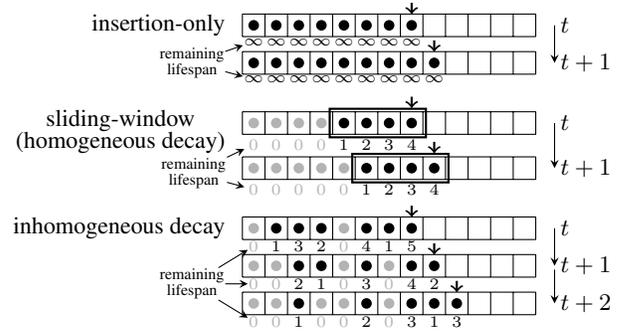
\begin{figure}[t]
\setlength{\belowcaptionskip}{-2ex}
\centering
\begin{tikzpicture}[
every node/.style={inner sep=0, font=\footnotesize},
item/.style={circle,fill,minimum size=4pt,anchor=center},
txt/.style={rectangle,inner sep=2pt, align=center},
sltxt/.style={txt,inner sep=1.5pt,font=\tiny},
stxt/.style={sltxt,inner sep=.2pt},
box/.style={rectangle,draw, inner sep=3.2pt, thick},
]

\draw[step=3mm] (0,0) grid (3.9,.3);
\coordinate(A1) at (.15,.15);
\coordinate[right=4 of A1](B1);
\foreach \i in {0,...,7}{
\node[right=.3*\i of A1, item] (i1\i) {};
\node[below=.08 of i1\i, stxt] (t1\i) {$\infty$};
}
\draw[thick,<-] (i17)++(0,.15) -- ++(0,.15);
\node[txt,anchor=west] at (B1) {$t$};

\begin{scope}[yshift=-5mm]
\draw[step=3mm] (0,0) grid (3.9,.3);
\coordinate(A2) at (.15,.15);
\coordinate[right=4 of A2](B2);
\foreach \i in {0,...,8} {
\node[right=.3*\i of A2, item] (i2\i) {};
\node[below=.08 of i2\i, stxt] (t2\i) {$\infty$};
}
\draw[thick,<-] (i28)++(0,.15) -- ++(0,.15);
\node[txt,anchor=west] at (B2) {$t+1$};
\end{scope}
\draw[-stealth] (B1) -- (B2);
\node[txt, left=.3 of A1] {insertion-only};
\node[stxt,left=.3 of i20] (lifespan) {remaining\\[-1pt]lifespan};
\draw[-stealth] (lifespan) -- (t10);
\draw[-stealth] (lifespan) -- (t20);

\begin{scope}[yshift=-1.3cm]
\draw[step=3mm] (0,0) grid (3.9,.3);
\coordinate(A1) at (.15,.15);
\coordinate[right=4 of A1](B1);
\foreach \i in {0,...,3} {
\node[right=.3*\i of A1, item, gray!60] (i1\i) {};
\node[below=.08 of i1\i, sltxt, gray!60] (t1\i) {$0$};
}
\foreach \i [count=\l] in {4,...,7} {
\node[right=.3*\i of A1, item] (i1\i) {};
\node[below=.08 of i1\i, sltxt] (t1\i) {$\l$};
}
\node[box,fit={(i14) (i17)}] {};
\draw[thick,<-] (i17)++(0,.2) -- ++(0,.15);
\node[txt,anchor=west] at (B1) {$t$};

\begin{scope}[yshift=-6mm]
\draw[step=3mm] (0,0) grid (3.9,.3);
\coordinate(A2) at (.15,.15);
\coordinate[right=4 of A2](B2);
\foreach \i in {0,...,4} {
\node[right=.3*\i of A2, item, gray!60] (i2\i) {};
\node[below=.08 of i2\i, sltxt, gray!60] (t2\i) {$0$};
}
\foreach \i [count=\l] in {5,...,8} {
\node[right=.3*\i of A2, item] (i2\i) {};
\node[below=.08 of i2\i, sltxt] (t2\i) {$\l$};
}
\node[box,fit={(i25) (i28)}] {};
\draw[thick,<-] (i28)++(0,.2) -- ++(0,.15);
\node[txt,anchor=west] at (B2) {$t+1$};
\end{scope}
\draw[-stealth] (B1) -- (B2);
\node[txt, below left=-.25 and .3 of A1] {sliding-window\\[-2pt](homogeneous decay)};

\node[stxt,below left=-.2 and .3 of i20] (lifespan) {remaining\\[-1pt]lifespan};
\draw[-stealth] (lifespan) -- (t10);
\draw[-stealth] (lifespan) -- (t20);
\end{scope}

\begin{scope}[yshift=-2.7cm]
\draw[step=3mm] (0,0) grid (3.9,.3);
\coordinate(A1) at (.15,.15);
\coordinate[right=4 of A1](B1);
\foreach \i in {0,4} {
\node[right=.3*\i of A1, item, gray!60] (i1\i) {};
\node[below=.08 of i1\i, stxt, gray!60] (t1\i) {$0$};
}
\foreach \i/\l in {1/1, 2/3, 3/2, 5/4, 6/1, 7/5} {
\node[right=.3*\i of A1, item] (i1\i) {};
\node[below=.08 of i1\i, stxt] (t1\i) {$\l$};
}
\node[txt,anchor=west] at (B1) {$t$};
\draw[thick,<-] (i17)++(0,.15) -- ++(0,.15);

\begin{scope}[yshift=-5mm]
\draw[step=3mm] (0,0) grid (3.9,.3);
\coordinate(A2) at (.15,.15);
\coordinate[right=4 of A2](B2);
\foreach \i in {0,1,4,6} {
\node[right=.3*\i of A2, item, gray!60] (i2\i) {};
\node[below=.08 of i2\i, stxt, gray!60] (t2\i) {$0$};
}
\foreach \i/\l in {2/2, 3/1, 5/3, 7/4, 8/2} {
\node[right=.3*\i of A2, item] (i2\i) {};
\node[below=.08 of i2\i, stxt] (t2\i) {$\l$};
}
\node[txt,anchor=west] at (B2) {$t+1$};
\draw[thick,<-] (i28)++(0,.15) -- ++(0,.15);
\end{scope}

\begin{scope}[yshift=-1cm]
\draw[step=3mm] (0,0) grid (3.9,.3);
\coordinate(A3) at (.15,.15);
\coordinate[right=4 of A3](B3);
\foreach \i in {0,1,3,4,6} {
\node[right=.3*\i of A3, item, gray!60] (i3\i) {};
\node[below=.08 of i3\i, stxt, gray!60] (t3\i) {$0$};
}
\foreach \i/\l in {2/1, 5/2, 7/3, 8/1, 9/3} {
\node[right=.3*\i of A3, item] (i3\i) {};
\node[below=.08 of i3\i, stxt] (t3\i) {$\l$};
}
\node[txt,anchor=west] at (B3) {$t+2$};
\draw[thick,<-] (i39)++(0,.15) -- ++(0,.15);
\end{scope}

\draw[-stealth] (B1) -- (B2);
\draw[-stealth] (B2)++(0,-.05) -- (B3);
\node[txt, left=.3 of A1] {inhomogeneous decay};
\node[stxt,left=.3 of t20] (lifespan) {remaining\\[-1pt]lifespan};
\draw[-stealth] (lifespan) -- (t10);
\draw[-stealth] (lifespan) -- (t20);
\draw[-stealth] (lifespan) -- (t30);
\end{scope}

\end{tikzpicture}

\caption{\textbf{Insertion-only stream}: each element has an infinite lifespan.
\textbf{Sliding-window stream}: each element has a same initial lifespan.
\textbf{Our model}: each element can have an arbitrary lifespan.}
\label{fig:illustration}
\end{figure}

We notice that these two existing streaming settings, i.e., insertion-only stream
and sliding-window stream, actually represent two extremes.
In insertion-only streams, a subset is selected from all historical data elements
which are treated as of equal importance, regardless of how outdated they are.
This is often undesirable because the stale historical data is usually less
important than fresh and recent data.
While in sliding-window streams, a subset is selected from the most recent data
only and historical data outside of the window is completely discarded.
This is also sometimes undesirable because one may not wish to completely lose the
entire history of past data and some historical data may be still important.
As a result, SSO over insertion-only streams may find solutions that are not
fresh; while SSO over sliding-window streams may find solutions that exclude
historical important data or include many recent but valueless data.
Can we design SSO techniques with a better streaming setting?

We observe that both insertion-only stream and sliding-window stream actually can
be unified by introducing the concept of data {\em lifespan}, which is the amount
of time an element participating in subset selection.
As time advances, an element's remaining lifespan decreases.
When an element's lifespan becomes zero, it is discarded and no longer
participates in subset selection.
Specifically, in insertion-only streams, each element has an infinite lifespan and
will always participate in subset selection after arrival.
While in sliding-window streams, each element has a same initial lifespan (i.e.,
the window size), and hence participates in subset selection for a same amount of
time (see Fig.~\ref{fig:illustration}).

We observe that in some real-world scenarios, it may be inappropriate to assume
that each element in a data stream has a same lifespan.
Let us consider the following scenario.

\header{Motivating Example.}
Consider a news aggregation website such as Hacker News~\cite{hn} where news
submitted by users form a news stream.
Interesting news may attract users to keep clicking and commenting and thus
survive for a long time; while boring news may only survive for one or two
days~\cite{Leskovec2009a}.
In news recommendation tasks, we should select a subset of news from {\em current
alive news} rather than the most recent news.

Therefore, besides timestamp of each data element, lifespan of each data element
should also be considered in subset selection.
Other similar scenarios include hot video selection from YouTube (where each video
may have its own lifespan), and trending hashtag selection from Twitter (where
each hashtag may have a different lifespan).

\header{Overview of Our Approach.}
We propose to extend the two extreme streaming settings to a more general
streaming setting where each element is allowed to have an arbitrary initial
lifespan and thus each element can participate in subset selection for an
arbitrary amount of time (see Fig.~\ref{fig:illustration}).
We refer to this more general decaying mechanism as {\em inhomogeneous decay}, in
contrast to the {\em homogeneous decay} adopted in sliding-window streams.
This work presents three algorithms to address SSO over streams with inhomogeneous
decays (SSO-ID).
We first present a simple streaming algorithm, i.e., \textsc{BasicStreaming}.
Then, we present \textsc{HistApprox} to improve the efficiency significantly.
Finally, we design a streaming version of \textsc{HistApprox}, i.e.,
\textsc{HistStreaming}.
We theoretically show that our algorithms have constant approximation factors.

Our main contributions include:
\begin{itemize}[itemsep=0pt,topsep=0pt]
\item We propose a general inhomogeneous-decaying streaming model that allows each
element to participate in subset selection for an arbitrary amount of time.
\item We design three algorithms to address the SSO-ID problem with constant
approximation factors.
\item We conduct experiments on real data, and the results demonstrate that our
method finds high quality solutions and is up to two orders of magnitude faster
than \textsc{Greedy}.
\end{itemize}

\section{Problem Statement}
\label{sec:preliminaries}

\header{Data Stream.}
A data stream comprises an unbounded sequence of elements arriving in
chronological order, denoted by $\{v_1,v_2,\ldots\}$.
Each element is from set $V$, called the {\em ground set}, and each element $v$
has a discrete timestamp $t_v\in\mathbb{N}$.
It is possible that multiple data elements arriving at the same time.
In addition, there may be other attributes associated with each element.

\header{Inhomogeneous Decay.}
We propose an {\em inhomogeneous-decaying data stream} (IDS) model to enable {\em
inhomogeneous decays}.
For an element $v$ arrived at time $t_v$, it is assigned an initial {\em lifespan}
$l(v,t_v)\in\mathbb{N}$ representing the maximum time span that the element will
remain active.
As time advances to $t\geq t_v$, the element's {\em remaining lifespan} decreases
to $l(v,t)\!\triangleq\! l(v,t_v)- (t-t_v)$.
If $l(v,t')=0$ at some time $t'$, $v$ is discarded.
We will assume $l(v,t_v)$ is given as an input to our algorithm.
At any time $t$, active elements in the stream form a set, denoted by
$\CS_t\!\triangleq\!\{v\colon v\in V\wedge t_v\leq t\wedge l(v,t)>0\}$.

IDS model is general.
If $l(v,t_v)=\infty,\forall v$, an IDS becomes an insertion-only stream.
If $l(v,t_v)=W,\forall v$, an IDS becomes a sliding-window stream.
If $l(v,t_v)$ follows a geometric distribution parameterized by $p$, i.e.,
$P(l(v,t_v)=l)=(1-p)^{l-1}p$, it is equivalent of saying that an active element is
discarded with probability $p$ at each time step.

To simplify notations, if time $t$ is clear from context, we will use $l_v$ to
represent $l(v,t)$, i.e., the remaining lifespan (or just say ``the lifespan'') of
element $v$ at time $t$.

\header{Monotone Submodular Function}~\cite{Nemhauser1978}.
A set function $f\colon 2^V\mapsto\mathbb{R}_{\geq 0}$ is submodular if
$f(S\cup\{s\})-f(S)\geq f(T\cup\{s\})-f(T)$, for all $S\subseteq T\subseteq V$ and
$s\in V\backslash T$.
$f$ is monotone (non-decreasing) if $f(S)\leq f(T)$ for all $S\subseteq T\subseteq
V$.
Without loss of generality, we assume $f$ is normalized, i.e., $f(\emptyset)=0$.

Let $\delta(s|S)\triangleq f(S\cup\{s\})-f(S)$ denote the \emph{marginal gain} of
adding element $s$ to $S$.
Then monotonicity is equivalent of saying that the marginal gain of every element
is always non-negative, and submodularity is equivalent of saying that marginal
gain $\delta(s|S)$ of element $s$ never increases as set $S$ grows bigger, aka the
diminishing returns property.

\header{Streaming Submodular Optimization with Inhomogeneous Decays (SSO-ID).}
Equipped with the above notations, we formulate the cardinality constrained SSO-ID
problem as follows:
\[
\OPT_t\triangleq\max_S f(S),\quad\text{s.t.}\quad S\subseteq\CS_t\wedge |S|\leq k,
\]
where $k$ is a given budget.

\remark.
The SSO-ID problem is NP-hard, and active data $\CS_t$ is continuously evolving
with outdated data being discarded and new data being added in at every time $t$,
which further complicates the algorithm design.
A naive algorithm to solve the SSO-ID problem is that, when $\CS_t$ is updated, we
re-run \textsc{Greedy} on $\CS_t$ from scratch, and this approach outputs a
solution that is $(1-1/e)$-approximate.
However, it needs $O(k|\CS_t|)$ utility function evaluations at each time step,
which is unaffordable for large $\CS_t$.
Our goal is to find faster algorithms with comparable approximation guarantees.

\section{Algorithms}
\label{sec:algorithms}

This section presents three algorithms to address the SSO-ID problem.
Due to space limitation, the proofs of all theorems are included in the extended
version of this paper.

\subsection{Warm-up: The \textsc{BasicStreaming} Algorithm}
\label{ss:basic_method}

In the literature, \textsc{SieveStreaming}~\cite{Badanidiyuru2014a} is designed to
address SSO over insertion-only streams.
We leverage \textsc{SieveStreaming} as a basic building block to design a
\textsc{BasicStreaming} algorithm.
\textsc{BasicStreaming} is simple per se and may be inefficient, but offers
opportunities for further improvement.
This section assumes lifespan is upper bounded by $L$, i.e., $l_v\leq L,\forall
v$.
We later remove this assumption in the following sections.

\header{\textsc{SieveStreaming}}~\cite{Badanidiyuru2014a} is a threshold based
streaming algorithm for solving cardinality constrained SSO over insertion-only
streams.
The high level idea is that, for each coming element, it is selected only if its
gain w.r.t.~a set is no less than a threshold.
In its implementation, \textsc{SieveStreaming} lazily maintains a set of
$\log_{1+\epsilon}2k=O(\epsilon^{-1}\log k)$ thresholds and each is associated
with a candidate set initialized empty.
For each coming element, its marginal gain w.r.t.~each candidate set is computed;
if the gain is no less than the corresponding threshold and the candidate set is
not full, the element is added in the candidate set.
At any time, a candidate set having the maximum utility is the current solution.
\textsc{SieveStreaming} achieves an $(1/2-\epsilon)$ approximation guarantee.

\header{Algorithm Description.}
We show how \textsc{SieveStreaming} can be used to design a
\textsc{BasicStreaming} algorithm to solve the SSO-ID problem.
Let $V_t$ denote a set of elements arrived at time $t$.
We partition $V_t$ into (at most) $L$ non-overlapping subsets, i.e.,
$V_t=\cup_{l=1}^L V_l^{(t)}$ where $V_l^{(t)}$ is the subset of elements with
lifespan $l$ at time $t$.
\textsc{BasicStreaming} maintains $L$ \textsc{SieveStreaming} instances, denoted
by $\{\CA_l^{(t)}\}_{l=1}^L$, and alternates a {\em data update} step and a {\em
time update} step to process the arriving elements $V_t$.

\bullethdr{Data Update.}
This step processes arriving data $V_t$.
Let instance $\CA_l^{(t)}$ only process elements with lifespan no less than $l$.
In other words, elements in $\cup_{i\geq l}V_i^{(t)}$ are fed to $\CA_l^{(t)}$.
After processing $V_t$, $\CA_1^{(t)}$ outputs the current solution.

\bullethdr{Time Update.}
This step prepares for processing the upcoming data in the next time step.
We reset instance $\CA_1^{(t)}$, i.e., empty its threshold set and each candidate
set.
Then we conduct a {\em circular shift} operation:
$\CA_1^{(t+1)}\!\gets\!\CA_2^{(t)},\CA_2^{(t+1)}\!\gets\!\CA_3^{(t)},
\ldots,\CA_L^{(t+1)}\!\gets\!\CA_1^{(t)}$.

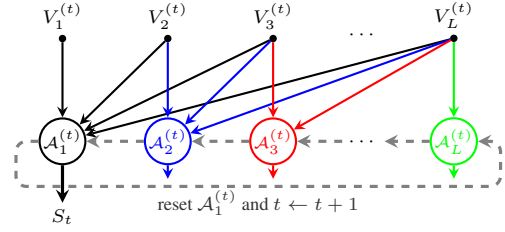
\begin{figure}[htp]
\centering

\begin{tikzpicture}[
>=stealth,
every node/.style={thick, inner sep=0, font=\scriptsize},
inputnode/.style={circle,draw,fill,minimum size=2pt},
alg/.style={circle, draw, minimum size=14pt, font=\tiny},
rarr/.style={<-, thick},
txt/.style={inner sep=2pt,align=center, black},
time_line/.style={->,dashed, very thick, gray},
]

\coordinate (c1) at(0,0) {};
\node[alg, right=1.4 of c1] (s1) {$\CA_1^{(t)}$};
\node[alg,blue,right=2.8 of c1] (s2) {$\CA_2^{(t)}$};
\node[alg,red,right=4.2 of c1] (s3) {$\CA_3^{(t)}$};
\node[right=.6 of s3, inner sep=2pt] (sd) {$\cdots$};
\node[alg,green,right=.6 of sd] (sL) {$\CA_L^{(t)}$};

\node[txt, below =.5 of s1] (temp) {$S_t$};
\draw[rarr, very thick] (temp) -- (s1);
\draw[->,thick,blue] (s2) -- ++(0, -5mm);
\draw[->,thick,red] (s3) -- ++(0, -5mm);
\draw[->,thick,green] (sL) -- ++(0, -5mm);

\foreach \i in {1,2,3,L}{
\node[inputnode,above=1 of s\i] (p\i) {};
\node[txt,above=0 of p\i] {$V_\i^{(t)}$};
}
\node[above=1.2 of sd] (pd) {$\cdots$};

\foreach \i in {1,2,3,L} \draw[rarr]        (s1) -- (p\i);
\foreach \i in {2,3,L}   \draw[blue, rarr]  (s2) -- (p\i);
\foreach \i in {3,L}     \draw[red, rarr]   (s3) -- (p\i);
\draw[green, rarr] (sL) -- (pL);

\foreach \x/\y in {L/d, d/3, 3/2, 2/1} \draw[time_line,->] (s\x) -- (s\y);
\coordinate[left = .3 of s1] (cnw);  \coordinate[right = .3 of sL] (cne);
\coordinate[below= .6 of cnw] (csw); \coordinate[below= .6 of cne] (cse);
\draw[time_line, rounded corners] (s1) -- (cnw) -- (csw) -- (cse)
node[black!80,below=1pt,midway] {reset $\CA_1^{(t)}$ and $t\gets t+1$}
-- (cne) -- (sL);

\end{tikzpicture}

\caption{\textsc{BasicStreaming}.
Solid lines denote data update, and dashed lines denote time update.}
\label{fig:basic}
\end{figure}

\textsc{BasicStreaming} alternates the two steps and continuously processes data
at each time step.
We illustrate \textsc{BasicStreaming} in Fig.~\ref{fig:basic}, with pseudo-code
given in Alg.~\ref{alg:basic}.

\begin{algorithm}[ht]
\KwIn{An IDS of data elements arriving over time}
\KwOut{A subset $S_t$ at any time $t$}
Initialize $L$ {\sc SieveStreaming} instances $\{\CA_l^{(1)}\}_{l=1}^L$\;
\For{$t=1,2,\ldots$}{
\For(\tcp*[f]{data update}){$l=1,\ldots,L$}{
Feed $\CA_l^{(t)}$ with data $\cup_{i\geq l}V_i^{(t)}$\;
}
$S_t\gets$ output of $\CA_1^{(t)}$\;
\lFor(\tcp*[f]{time update}){$l=2,\ldots,L$}{$\CA_{l-1}^{(t+1)}\gets\CA_l^{(t)}$}
Reset $\CA_1^{(t)}$ and $\CA_L^{(t+1)}\gets\CA_1^{(t)}$\;
}
\caption{\textsc{BasicStreaming}}
\label{alg:basic}
\end{algorithm}

\header{Analysis.}
\textsc{BasicStreaming} exhibits a feature that an instance gradually expires (and
is reset) as data processed in it expires.
Such a feature ensures that, \textbf{\em at any time $t$, $\CA_1^{(t)}$ always
processed all the data in $\CS_t$}.
Because $\CA_1^{(t)}$ is a \textsc{SieveStreaming} instance, we immediately have
the following conclusions.

\begin{theorem}
\textsc{BasicStreaming} achieves an $(1/2-\epsilon)$ approximation guarantee.
\end{theorem}

\begin{theorem}
{\sc BasicStreaming} uses $O(L\epsilon^{-1}\log k)$ time to process each
element, and $O(Lk\epsilon^{-1}\log k)$ memory to store intermediate results
(i.e., candidate sets).
\end{theorem}

\remark.
As illustrated in Fig.~\ref{fig:basic}, data with lifespan $l$ will be fed to
$\{\CA_i^{(t)}\}_{i\leq l}$.
Hence, elements with large lifespans will fan out to a large fraction of {\sc
SieveStreaming} instances, and incur high CPU and memory usage, especially when
$L$ is large.
This is the main \textbf{bottleneck} of {\sc BasicStreaming}.
On the other hand, elements with small lifespans only need to be fed to a few
instances.
Therefore, if data lifespans are mainly distributed over small values, e.g.,
power-law distributed, then {\sc BasicStreaming} is still efficient.

\subsection{\textsc{HistApprox}: Improving Efficiency}
\label{ss:hist}

To address the bottleneck of \textsc{BasicStreaming} when processing data with a
large lifespan, we design \textsc{HistApprox} in this section.
\textsc{HistApprox} can significantly improve the efficiency of
\textsc{BasicStreaming} but requires active data $\CS_t$ to be stored in
RAM\footnote{For example, if lifespan follows a geometric distribution, i.e.,
$P(l_v=l)=(1-p)p^{l-1}, l=1,2,\ldots$, and at most $M$ elements arrive at a
time, then $|\CS_t|\leq\sum_{a=0}^{t-1}Mp^a\leq\frac{M}{1-p}$.
Hence, if RAM is larger than $\frac{M}{1-p}$, $\CS_t$ actually can be stored in
RAM even as $t\rightarrow\infty$.}.
Strictly speaking, \textsc{HistApprox} is not a streaming algorithm.
We later remove the assumption of storing $\CS_t$ in RAM in the next section.

\header{Basic Idea.}
If at any time, only a few instances are maintained and running in
\textsc{BasicStreaming}, then both CPU time and memory usage will decrease.
Our idea is hence to selectively maintain a subset of \textsc{SieveStreaming}
instances that can approximate the rest.
Roughly speaking, this idea can be thought of as using a histogram to approximate
a curve.
Specifically, let $g_t(l)$ denote the value of output of $\CA_l^{(t)}$ at time
$t$.
For very large $L$, we can think of $\{g_t(l)\}_{l\geq 1}$ as a ``curve'' (e.g.,
the dashed curve in Fig.~\ref{fig:histogram}).
Our idea is to pick a few instances as {\em active instances} and construct a
histogram to approximate this curve, as illustrated in Fig.~\ref{fig:histogram}.

\begin{figure}[htp]
\centering
\begin{tikzpicture}[
every node/.style={inner sep=0,font=\footnotesize},
vline/.style={dashed},
hline/.style={very thick},
txt/.style={rectangle,inner sep=2pt},
ins_txt/.style={rectangle,inner sep=1pt,anchor=south},
circ/.style={circle, draw, fill=white, thick, minimum size=2.5pt},
disk/.style={circle, draw, fill, thick, minimum size=2.5pt},
insertpoint/.style={stealth-,very thick, red},
]

\draw[very thick,blue,dashed] (0,1.4) .. controls (.5,1.5) and (.8,.8)
.. (1.5,.9) .. controls (2.4,1.1)
.. (3.5,.6) .. controls (4.5,.15) and (5.5,0)
.. (7,0);

\draw (0,0) -- (7,0);
\draw (0,0) -- ++(0,.1);

\draw[vline] (.5,0) node[txt,anchor=north] {$x_1$} -- (.5,1.25);
\draw[vline] (2.5,0) node[txt,anchor=north] {$x_2$} -- (2.5,1);
\draw[vline] (3.15,0) node[txt,anchor=north] {$x_3$} -- (3.15,.75);
\draw[vline] (3.75,0) node[txt,anchor=north] {$x_4$} -- (3.75,.5);
\draw[vline] (4.5,0) node[txt,anchor=north] {$x_5$} -- (4.5,.25);

\node[txt,inner sep=.5pt, anchor=north east] at (0,0) {$1$};

\draw[hline] (0,1.25) -- (.5,1.25) node[disk]{};
\draw[hline] (.5,1) node[circ]{} -- (2.5,1) node[disk]{};
\draw[hline] (2.5,.75) node[circ]{} -- (3.15,.75) node[disk]{};
\draw[hline] (3.15,.5) node[circ]{} -- (3.75,.5) node[disk]{};
\draw[hline] (3.75,.25) node[circ]{} -- (4.5,.25) node[disk]{};

\draw[insertpoint] (3.15,.85) -- ++(0,1.5ex) node[ins_txt] {Case 1};
\draw[insertpoint] (5.5,.15) -- ++(0,1.5ex) node[ins_txt] {Case 2};
\draw[insertpoint] (4.3,.4) -- ++(0,1.5ex) node[ins_txt] {Case 3};

\draw[-latex,blue] (1,1.35) node[anchor=west,blue]{$\{g_t(l)\}_{l\geq 1}$}
to[out=160,in=45] (.7,1.2);

\draw[-latex] (4,1.2) node[anchor=west]{$\{g_t(l)\}_{l\in\bx_t}$}
to[out=180,in=80] (3.3,.51);

\end{tikzpicture}

\caption{Approximate $\{g_t(l)\}_{l\geq 1}$ by $\{g_t(l)\}_{l\in\bx_t}$.}
\label{fig:histogram}
\end{figure}
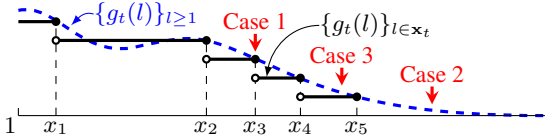

The challenge is that, as new data keeps arriving, the curve is changing; hence,
we need to update the histogram accordingly to make sure that the histogram always
well approximates the curve.
Let $\bx_t\triangleq\{x_1^{(t)},x_2^{(t)},\ldots\}$ index a set of active
instances at time $t$, where each index $x_i^{(t)}\geq 1$.\footnote{Superscript
$t$ will be omitted if time $t$ is clear from context.}
In the follows, we describe the $\bx_t$ updating method, i.e.,
\textsc{HistApprox}, and the method guarantees that the maintained histogram
satisfies our requirement.

\header{Algorithm Description.}
\textsc{HistApprox} consists of two steps: (1) updating indices; (2) removing
redundant indices.

\bullethdr{Updating Indices.}
The algorithm starts with an empty index set, i.e., $\bx_1=\emptyset$.
At time $t$, consider a set of newly arrived elements $V_l^{(t)}$ with lifespan
$l$.
These elements will increase the curve before $l$ (because data $V_l^{(t)}$ will
be fed to $\{\CA_i^{(t)}\}_{i\leq l}$, see Fig.~\ref{fig:basic}).
There are three cases based on the position of $l$, as illustrated in
Fig.~\ref{fig:histogram}.

Case 1.
If $l\in\bx_t$, we simply feed $V_l^{(t)}$ to $\{\CA_i^{(t)}\}_{i\in\bx_t\wedge
i\leq l}$.

Case 2.
If $l\notin\bx_t$ and $l$ has no successor in $\bx_t$, we create a new instance
$\CA_l^{(t)}$ and feed $V_l^{(t)}$ to $\{\CA_i^{(t)}\}_{i\in\bx_t\wedge i\leq l}$.

Case 3.
If $l\notin\bx_t$ and $l$ has a successor $l_2\in\bx_t$.
Let $\CA_l^{(t)}$ be a copy of $\CA_{l_2}^{(t)}$, then we feed $V_l^{(t)}$ to
$\{\CA_i^{(t)}\}_{i\in\bx_t\wedge i\leq l}$.
Note that $\CA_l^{(t)}$ needs to process all data with lifespan $\geq l$ at time
$t$.
Because $\CA_{l_2}^{(t)}$ has processed all data with lifespan $\geq l_2$, we
still need to feed $\CA_l^{(t)}$ with historical data s.t.~their lifespan $\in
[l,l_2)$.
That is the reason we need $\CS_t$ to be stored in RAM.

Above scheme guarantees that each $\CA_l^{(t)}, l\in\bx_t$ processed all the data
with lifespan $\geq l$ at time $t$.
The detailed pseudo-code is given in procedure \Process of Alg.~\ref{alg:hist}.

\bullethdr{Removing Redundant Indices.}
Intuitively, if the outputs of two instances are close to each other, it is not
necessary to keep both of them.
We need the following definition to quantify redundancy.

\begin{definition}[$\epsilon$-redundancy]
At time $t$, consider two instances $\CA_i^{(t)}$ and $\CA_l^{(t)}$ with $i<l$.
We say $\CA_l^{(t)}$ is $\epsilon$-redundant if their exists $j>l$ such that
$g_t(j)\geq (1-\epsilon)g_t(i)$.
\end{definition}

The above definition simply states that, since $\CA_i^{(t)}$ and $\CA_j^{(t)}$ are
already close with each other, then instances between them are redundant.
In \textsc{HistApprox}, we regularly check the output of each instance and
terminate those redundant ones, as described in \ReduceRedundancy of
Alg.~\ref{alg:hist}.

\begin{algorithm}[t]
\KwIn{An IDS of data elements arriving over time}
\KwOut{A subset $S_t$ at any time $t$}
$\bx_1\gets\emptyset$\;
\For(\tcp*[f]{$V_t=\cup_lV_l^{(t)}$}){$t=1,2,\ldots$}{
\lForEach(\tcp*[f]{data update}){$V_l^{(t)}\neq\emptyset$}{%
\Process{$V_l^{(t)}$}}
$S_t\gets$ output of $\CA_{x_1}^{(t)}$\;
\lIf(\tcp*[f]{time update}){$x_1\!=\!1$}{%
Kill $\CA_1^{(t)}$, $\bx_t\!\gets\!\bx_t\backslash\{1\}$}
\For{$i=1,\ldots,|\bx_t|$}{
$\CA^{(t+1)}_{x_i-1}\gets\CA^{(t)}_{x_i}$,
$x_i^{(t+1)}\gets x_i^{(t)}-1$\;
}
}

\myproc{\Process{$V_l^{(t)}$}}{\label{ln:process_edges}
\If{$l\notin\bx_t$}{
\uIf(\tcp*[f]{Case 2 in Fig.~\ref{fig:histogram}})
{$l$ has no successor in $\bx_t$}{
$\CA_l^{(t)}\gets$ new instance\;
}
\Else(\tcp*[h]{let $l_2$ denote the successor of $l$}){
$\CA_l^{(t)}\gets$ a copy of $\CA_{l_2}^{(t)}$\tcp*{Case 3 in
Fig.~\ref{fig:histogram}}\label{ln:copy}
Feed $\CA_l^{(t)}$ with historical data elements s.t. their
lifespans $\in[l,l_2)$\;\label{ln:process_history}
}
$\bx_t\gets\bx_t\cup\{l\}$\;
}
\lForEach{$i\in\bx_t$ and $i\leq l$}{Feed $\CA_i^{(t)}$ with $V_l^{(t)}$}\label{ln:process_new}
\ReduceRedundancy{}\;
}

\myproc{\ReduceRedundancy{}}{\label{ln:remove_redundancy}
\ForEach{$i\in\bx_t$}{
Find the largest $j\!>\!i$ in $\bx_t$ s.t.
$g_t(j)\!\geq\! (1\!-\!\epsilon)g_t(i)$\;\label{ln:cond}
Delete each index $l\!\in\!\bx_t$ s.t.
$i\!<\!l\!<\!j$ and kill $\CA_l^{(t)}$\;\label{ln:remove}
}
}
\caption{\textsc{HistApprox}}
\label{alg:hist}
\end{algorithm}

\header{Analysis.}
Notice that indices $x\in\bx_t$ and $x+1\in\bx_{t-1}$ are actually the same index
(if they both exist) but appear at different time.
In general, we say $x'\in\bx_{t'}$ is an {\em ancestor} of $x\in\bx_t$ if $t'\leq
t$ and $x'=x+t-t'$.
In the follows, let $x'$ denote $x$'s ancestor at time $t'$.
First, \textsc{HistApprox} maintains a histogram satisfying the following
property.

\begin{lemma}\label{lem:histogram_property}
For two consecutive indices $x_i,x_{i+1}\in\bx_t$ at any time $t$, one of the
following two cases holds:
\begin{description}
\item[C1] $\CS_t$ contains no data with lifespan $\in(x_i, x_{i+1})$.
\item[C2] $g_{t'}(x_{i+1}')\geq (1-\epsilon)g_{t'}(x_i')$ at some time $t'\leq
t$, and from time $t'$ to $t$, there is no data with lifespan between the two
indices arrived (exclusive).
\end{description}
\end{lemma}

Histogram with property C2 is known as a {\em smooth
histogram}~\cite{Braverman2007}.
Smooth histogram together with the submodularity of $f$ are sufficient to ensure a
constant factor approximation guarantee of $g_t(x_1)$.

\begin{theorem}\label{thm:HistApprox_guarantee}
\textsc{HistApprox} is $(1/3 - \epsilon)$-approximate, i.e., at any time $t$,
$g_t(x_1)\geq (1/3-\epsilon)\OPT_t$.
\end{theorem}

\begin{theorem}\label{thm:HistApprox_complexity}
\textsc{HistApprox} uses $O(\epsilon^{-2}\log^2 k)$ time to process each coming
element and $O(k\epsilon^{-2}\log^2k)$ memory to store intermediate results and
$|\CS_t|$ memory to store $\CS_t$.
\end{theorem}

\remark.
Because we use a histogram to approximate a curve, \textsc{HistApprox} has a
weaker approximation guarantee than \textsc{BasicStreaming}.
In experiments, we observe that \textsc{HistApprox} finds solutions with quality
very close to \textsc{BasicStreaming} and is much faster.
The main drawback of \textsc{HistApprox} is that active data $\CS_t$ needs to be
stored in RAM to ensure each $\CA_l^{(t)}$'s output is accurate.
If $\CS_t$ is larger than RAM capacity, then \textsc{HistApprox} is inapplicable.
We address this limitation in the following section.

\subsection{\textsc{HistStreaming}: A Heuristic Streaming Algorithm}
\label{ss:hist_streaming}

Based on \textsc{HistApprox}, this section presents a streaming algorithm
\textsc{HistStreaming}, which uses heuristics to further improve the efficiency of
\textsc{HistApprox}.
\textsc{HistStreaming} no longer requires storing active data $\CS_t$ in memory.

\header{Basic Idea.}
If we do not need to process the historical data in \textsc{HistApprox}
(Line~\ref{ln:process_history}), then there is no need to store $\CS_t$.
What if $\CA_l^{(t)}$ does not process historical data?
Because $\CA_l^{(t)}$ does not process all the data with lifespan $\geq l$ in
$\CS_t$, there will be a bias between its actual output $\hat{g}_t(l)$ and
expected output $g_t(l)$.
We only need to worry about the case $\hat{g}_t(l)<g_t(l)$, as the other case
$\hat{g}_t(l)\geq g_t(l)$ means that without processing historical data,
$\CA_l^{(t)}$ finds even better solutions (which may rarely happen in practice but
indeed possible).
In the follows, we apply two useful heuristics to design \textsc{HistStreaming},
and show that historical data can be ignored due to its insignificance and
submodularity of objective function.

\header{Effects of historical data.}
Intuitively, if historical data is insignificant, then a \textsc{SieveStreaming}
instance may not need to process it at all, and can still output quality
guaranteed solutions.
We notice that, in \textsc{HistApprox}, a newly created instance $\CA_l^{(t)}$
essentially needs to process three substreams: (1) elements arrived before $t$
with lifespan $\leq l_2$ (Line~\ref{ln:copy})\footnote{This substream is actually
processed by $\CA_l^{(t)}$'s successor $\CA_{l_2}^{(t)}$, and note that
$\CA_l^{(t)}$ is copied from $\CA_{l_2}^{(t)}$.}; (2) unprocessed historical
elements with lifespan $\in [l,l_2)$ (Line~\ref{ln:process_history}); (3) newly
arrived elements $V_l$ (Line~\ref{ln:process_new}).
Denote these three substreams by $S_1,S_2$ and $S_3$, respectively.
We state a useful lemma below.

\begin{lemma}\label{lem:substreams}
Let $S_1\Vert S_2\Vert S_3$ denote the concatenation of three substreams
$S_1,S_2,S_3$.
Let $\CA(S)$ denote the output value of applying \textsc{SieveStreaming}
algorithm $\CA$ on stream $S$.
If $\CA(S_1)\geq\alpha\CA(S_1\Vert S_2)$ for $0<\alpha<1$, then $\CA(S_1\Vert
S_3)\geq(1/4-\epsilon)\alpha\OPT$ where $\OPT$ is the value of an optimal
solution in stream $S_1\Vert S_2\Vert S_3$.
\end{lemma}

Lemma~\ref{lem:substreams} states that, if historical data $S_2$ is insignificant,
i.e., $\CA(S_1)\geq\alpha\CA(S_1\Vert S_2)$ for $0<\alpha<1$ (the closer $\alpha$
is to $1$, the less significant $S_2$ is), then an instance does not need to
process $S_2$ and still finds quality guaranteed solutions.
This will further ensure that \textsc{HistApprox} finds quality guaranteed
solutions (more explanation on this point can be found in the extended version of
this paper).
Although it is intractable to theoretically show that historical data $S_2$ is
indeed insignificant, intuitively, as unprocessed historical data $S_2$ is caused
by the deletion of redundant instances (consider the example given in
Fig.~\ref{fig:historical_data}).
These instances are redundant because $\CA(S_1||S_2)$ does not increase much upon
$\CA(S_1)$.
Hence, it makes sense to assume that historical data $S_2$ is insignificant, and
by Lemma~\ref{lem:substreams}, $S_2$ can be ignored.

\begin{figure}[htp]
\centering
\begin{tikzpicture}[
every node/.style={inner sep=0pt},
txt/.style={inner sep=2pt, anchor=north, font=\footnotesize},
vline/.style={densely dashed},
hline/.style={very thick},
gline/.style={thick,densely dotted,gray},
circ/.style={circle, draw, fill=white, thick, minimum size=2.5pt},
disk/.style={circle, draw, fill, thick, minimum size=2.5pt},
add/.style={stealth-,thick,red},
change/.style={-latex, thick, blue,double},
]

\draw (0,0) -- (3,0);

\draw[vline] (.5,0) node[txt] {$l_1'$} -- (.5,.5);
\draw (2,0) node[txt,font=\scriptsize] {$l_0'$} [vline,gray!60] -- (2,.3);
\draw[vline] (2.5,0) node[txt] {$l_2'$} --  (2.5,.2);

\draw[hline,red] (0,.56) -- (.5,.56) node[disk] {};
\draw[hline] (.5,.2) node[circ] {} -- (2.5,.2) node[disk] {};
\draw[hline] (2.5,0) node[circ] {} -- (3,0);

\draw[gline] (.5,.28) -- (2,.28);
\node[gray!60, circ] at (.5,.28) {};
\node[disk, gray!60] at (2,.28) {};

\draw[add] (2,.45) -- ++(0,6pt);

\begin{scope}[xshift=4.5cm]
\draw (0,0) -- (3,0);

\draw[vline] (.5,0) node[txt] {$l_1$} -- (.5,.55);
\draw[vline] (1.5,0) node[txt] {$l$} -- (1.5,.35);
\draw[vline] (2.5,.2) -- (2.5,0) node[txt] {$l_2$};

\draw[hline,red] (0,.62) -- (.5,.62) node[disk] {};
\draw[hline,red] (.5,.35) node[circ] {} -- (1.5,.35) node[disk] {};
\draw[hline] (1.5,.2) node[circ] {} -- (2.5,.2) node[disk] {};
\draw[hline] (2.5,0) node[circ] {} -- (3,0);

\draw[gray] (2,.1) -- ++(0,-.1) node[txt,font=\scriptsize] {$l_0$};

\draw[add] (1.5,.5) -- ++(0,6pt);
\end{scope}

\draw[->] (3.3,.4) node[txt,anchor=south] {$t'$} -- ++(.8,0)
node[txt,anchor=south] {$t$};
\end{tikzpicture}

\caption{At time $t'$, data with lifespan $l_0'$ arrives and forms a redundant
instance, which is removed.
At time $t>t'$, data with lifespan $l$ arrives and $\CA_l^{(t)}$ is created.
Data at $l_0$ becomes the unprocessed historical data.
We thus say that unprocessed historical data is caused by the deletion of
redundant instance at previous time.}
\label{fig:historical_data}
\end{figure}
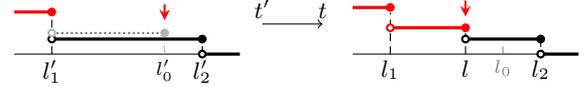

\header{Protecting non-redundant instances.}
To further ensure the solution quality of \textsc{HistStreaming}, we introduce
another heuristic to protect non-redundant instances.

Because $g_t(l)$ is unknown, to avoid removing instances that are actually not
redundant, we give each instance $\CA_l^{(t)}$ an amount of value, denoted by
$\delta_l$, as compensation for not processing historical data, i.e., $g_t(l)$ may
be as large as $\hat{g}_t(l)+\delta_l$.
This allows us to represent $g_t(l)$ by an interval
$[\underline{g}_t(l),\overline{g}_t(l)]$ where
$\underline{g}_t(l)\triangleq\hat{g}_t(l)$ and
$\overline{g}_t(l)\triangleq\hat{g}_t(l)+\delta_l$.
As $\underline{g}_t(j)\geq (1-\epsilon)\overline{g}_t(i)$ implies $g_t(j)\geq
(1-\epsilon)g_t(i)$, the condition in Line~\ref{ln:cond} of \textsc{HistApprox}
is replaced by $\underline{g}_t(j)\geq (1-\epsilon)\overline{g}_t(i)$.

We want $\delta_l$ to be related to the amount of historical data that
$\CA_l^{(t)}$ does not process.
Recall the example in Fig.~\ref{fig:historical_data}.
Unprocessed historical data is always fed to $l$'s predecessor instance whenever
redundant instances are removed in the interval $l$ belonging to.
Also notice that $g_{t'}(l_2')\geq(1-\epsilon)g_{t'}(l_1')$ holds after the
removal of redundant instances.
Hence, the contribution of unprocessed historical data can be estimated to be at
most $\epsilon g_{t'}(l_1')$.
In general, if some redundant indices are removed in interval $(i,j)$ at time $t$,
we set $\delta_l=\epsilon\overline{g}_t(i)$ for index $l$ that is later created in
the interval $(i,j)$.

\header{Algorithm Description.}
We only need to slightly modify \Process and \ReduceRedundancy (see
Alg.~\ref{alg:HistStreaming}).

\begin{algorithm}[t]
\myproc{\Process{$V_l^{(t)}$}}{
\If{$l\notin\bx_t$}{
$\delta_l\gets 0$\;
$\cdots$\\
\tcp{If $l$ has a successor $l_2$}
$\CA_l^{(t)}\gets$ a copy of $\CA_{l_2}^{(t)}$\;
Find $i,j\in\bx_t$ s.t. $l\in (i,j)$ and $\delta_{ij}$ is recorded, then let
$\delta_l\gets\delta_{ij}$\;
}
$\cdots$\\
}

\myproc{\ReduceRedundancy{}}{
\ForEach{$i\in\bx_t$}{
Find the largest $j\!>\!i$ in $\bx_t$ s.t.
$\underline{g}_t(j)\!\geq\!(1\!-\!\epsilon)\overline{g}_t(i)$\;
Delete each index $l\!\in\!\bx_t$ s.t.
$i\!<\!l\!<\!j$ and kill $\CA_l^{(t)}$\;
\tcp{Record the amount of unprocessed data in $(i,j)$}
$\delta_{ij}\gets\epsilon\overline{g}_t(i)$\;
}
}
\caption{\textsc{HistStreaming}}
\label{alg:HistStreaming}
\end{algorithm}

\remark.
\textsc{HistStreaming} uses heuristics to further improve the efficiency of
\textsc{HistApprox}, and no longer needs to store $\CS_t$ in memory.
In experiments, we observe that \textsc{HistStreaming} can find high quality
solutions.

\section{Experiments}
\label{sec:experiments}

In this section, we construct several maximum coverage problems to evaluate the
performance of our methods.
We use real world and public available datasets.
Note that the optimization problems defined on these datasets may seem to be
simplistic, as our main purpose is to validate the performance of proposed
algorithms, and hence we want to keep the problem settings as simple and clear as
possible.

\subsection{Datasets}

\header{DBLP.}
We construct a \emph{representative author selection} problem on the DBLP
dataset~\cite{dblp}, which records the meta information of about $3$ million
papers, including $1.8$ million authors and $5,079$ conferences from 1936 to 2018.
We say that an author represents a conference if the author published papers in
the conference.
Our goal is to maintain a small set of $k$ authors that jointly represent the
maximum number of distinct conferences at any time.
We filter out authors that published less than $10$ papers and sort the remaining
$188,383$ authors by their first publication date to form an author stream.
On this dataset, an author's lifespan could be defined as the time period between
its first and last publication dates.

\header{StackExchange.}
We construct a \emph{hot question selection} problem on the math.stackexchange.com
website~\cite{stackexchange}.
The dataset records about $1.3$ million questions with $152$ thousand commenters
from $7/2010$ to $6/2018$.
We say a question is hot if it attracts many commenters to comment.
Our goal is select a small set of $k$ questions that jointly attract the maximum
number of distinct commenters at any time.
The questions are ordered by the post date, and the lifespan of a question can be
defined as the time interval length between its post time and last comment time.

\subsection{Settings}

\header{Benchmarks.} We consider the following two methods as benchmarks.
\begin{itemize}
\item\textbf{\textsc{Greedy}}.
We re-run \textsc{Greedy} on the active data $\CS_t$ at each time $t$, and apply
the {\em lazy evaluation} trick~\cite{Minoux1978} to further improve its
efficiency.
\textsc{Greedy} will serve as an upper bound.
\item\textbf{\textsc{Random}}.
We randomly pick $k$ elements from active data $\CS_t$ at each time $t$.
\textsc{Random} will serve as a lower bound.
\end{itemize}

\header{Efficiency Measure.}
When evaluating algorithm efficiency, we follow the previous
work~\cite{Badanidiyuru2014a} and record the number of utility function
evaluations, i.e., the number of {\em oracle calls}.
The advantage of this measure is that it is independent of the concrete algorithm
implementation and platform.

\header{Lifespan Generating.}
In order to test the algorithm performance under different lifespan distributions,
we also consider generating data lifespans by sampling from a geometric
distribution, i.e., $P(l_e=l)=(1-p)^{l-1}p, l=1,2,\ldots$.
Here $0<p<1$ controls the skewness of geometric distribution, i.e., larger $p$
implies that a data element is more likely to have a small lifespan.

\subsection{Results}

\header{Analyzing \textsc{BasicStreaming}.}
Before comparing the performance of our algorithms with benchmarks, let us first
study the properties of \textsc{BasicStreaming}, as it is the basis of
\textsc{HistApprox} and \textsc{HistStreaming}.
We mainly analyze how lifespan distribution affects the performance of
\textsc{BasicStreaming}.
To this end, we generate lifespans from $\mathit{Geo}(p)$ with varying $p$, and
truncate the lifespan at $L=1,000$.
We run the three proposed algorithms for $100$ time steps and maintain a set with
cardinality $k=10$ at every time step.
We set $\epsilon=0.1$.
The solution value and number of oracle calls (at time $t=100$) are depicted in
Figs.~\ref{fig:BasicStreaming_value} and~\ref{fig:BasicStreaming_calls},
respectively.

\begin{figure}[htp]
\centering
\subfloat[DBLP]{%
\includegraphics[width=.5\linewidth]{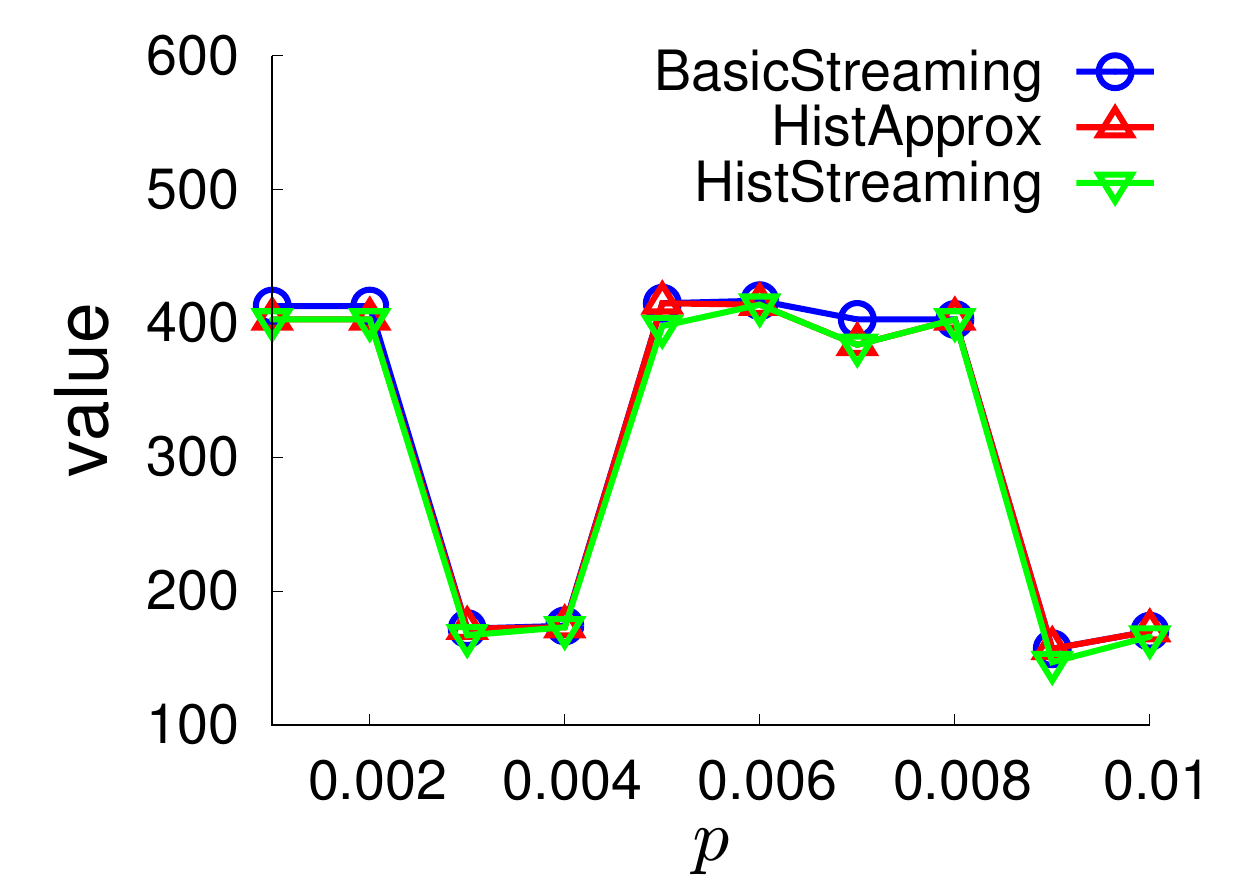}}
\subfloat[StackExchange]{%
\includegraphics[width=.5\linewidth]{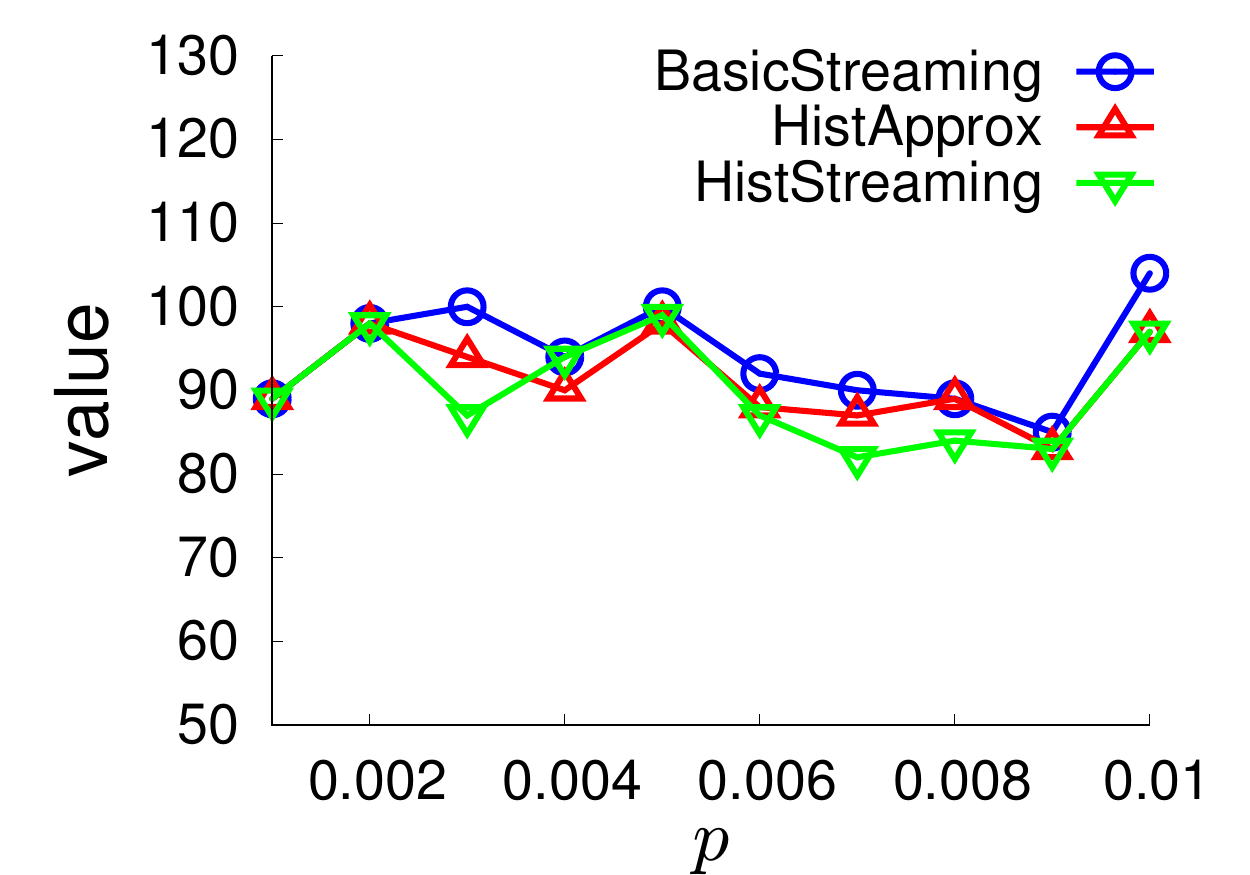}}
\caption{Solution value comparison (higher is better)}
\label{fig:BasicStreaming_value}
\end{figure}

\begin{figure}[htp]
\centering
\subfloat[DBLP]{%
\includegraphics[width=.5\linewidth]{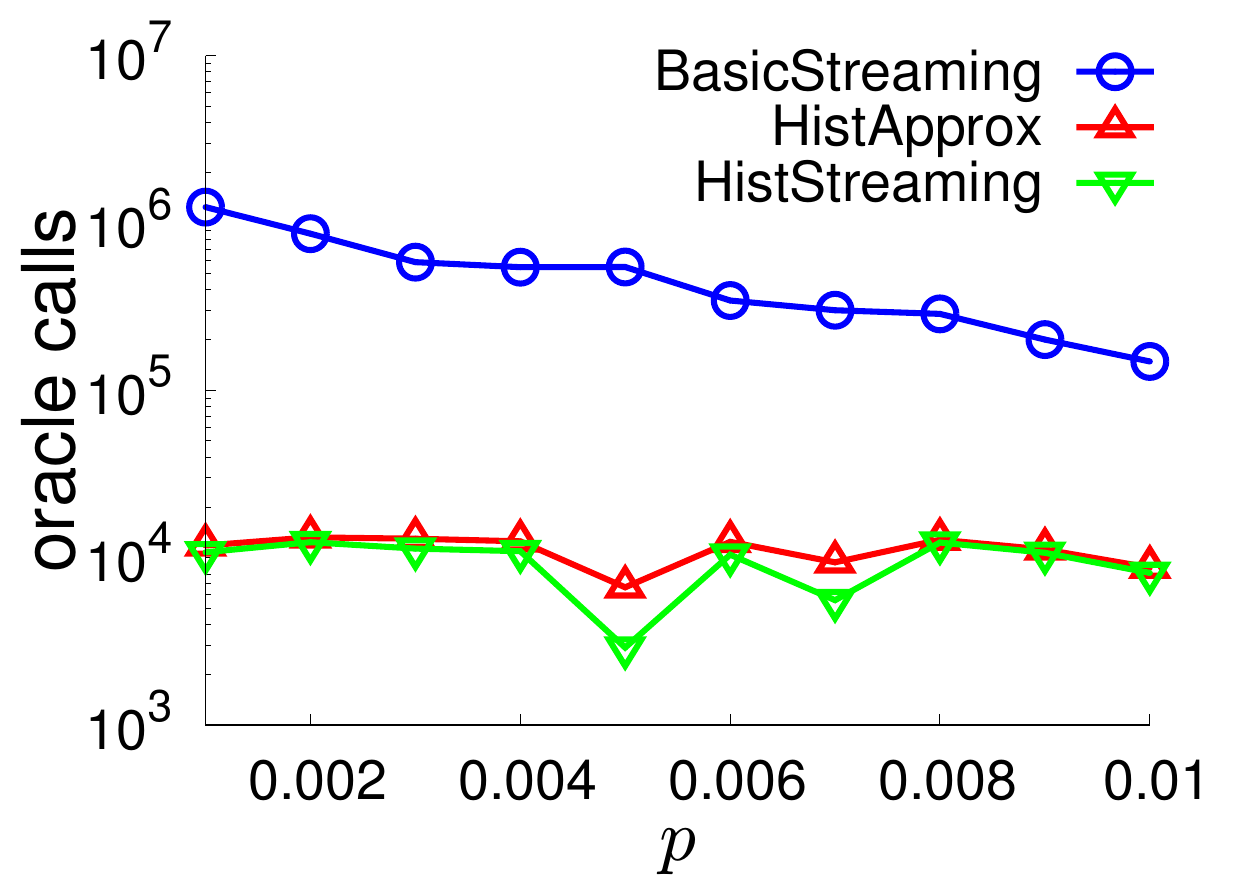}}
\subfloat[StackExchange]{%
\includegraphics[width=.5\linewidth]{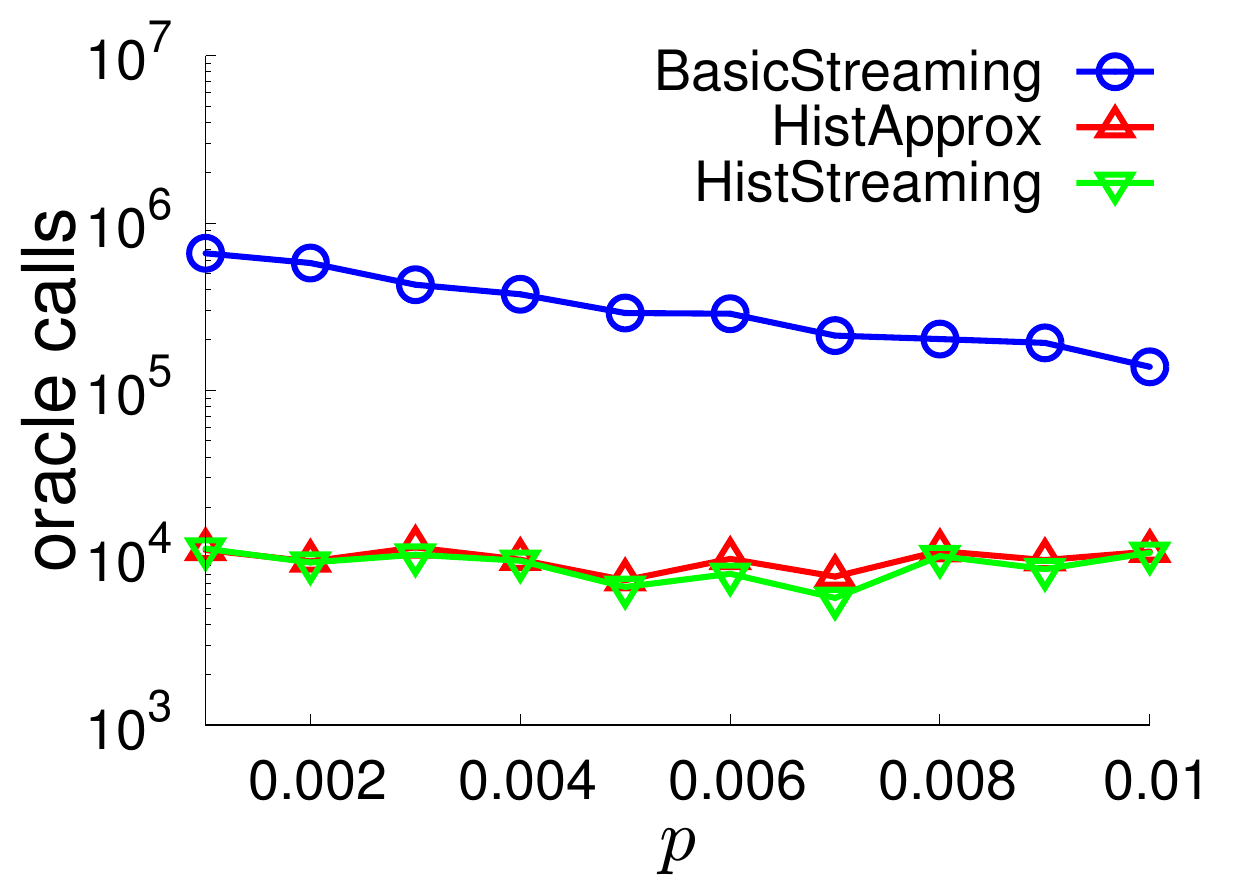}}
\caption{Oracle calls comparison (lower is better)}
\label{fig:BasicStreaming_calls}
\end{figure}

Figure~\ref{fig:BasicStreaming_value} states that the outputs of the three methods
are always close with each other under different lifespan distributions, i.e.,
they always output similar quality solutions.
Fig.~\ref{fig:BasicStreaming_calls} states that \textsc{BasicStreaming} requires
much more oracle calls than the other two methods, indicating that
\textsc{BasicStreaming} is less efficient than the other two methods.
We also observe that, as $p$ increases (hence more data elements tend to have
small lifespans), the number of oracle calls of \textsc{BasicStreaming} decreases.
This confirms our previous analysis that \textsc{BasicStreaming} is efficient when
most data elements have small lifespans.
We also observe that \textsc{HistApprox} and \textsc{HistStreaming} are not quite
sensitive to lifetime distribution, and they are much more efficient than
\textsc{BasicStreaming}.
In addition, we observe that \textsc{HistStreaming} is slightly faster than
\textsc{HistApprox} even though \textsc{HistStreaming} uses smaller RAM.

This experiment demonstrates that \textsc{BasicStreaming}, \textsc{HistApprox}, and
\textsc{HistStreaming} find solutions with similar quality, but
\textsc{HistApprox} and \textsc{HistStreaming} are much more efficient than
\textsc{BasicStreaming}.

\header{Performance Over Time.}
In the next experiment, we focus on analyzing the performance of
\textsc{HistStreaming}.
We fix the lifespan distribution to be $\mathit{Geo}(0.001)$ with $L=10,000$, and
run each method for $5,000$ time steps to maintain a set with cardinality $k=10$.
Figs.~\ref{fig:val} and~\ref{fig:calls} depict the solution value and ratio of the
number of oracle calls (w.r.t.~\textsc{Greedy}), respectively.

\begin{figure}[htp]
\centering
\subfloat[DBLP]{\includegraphics[width=.5\linewidth]{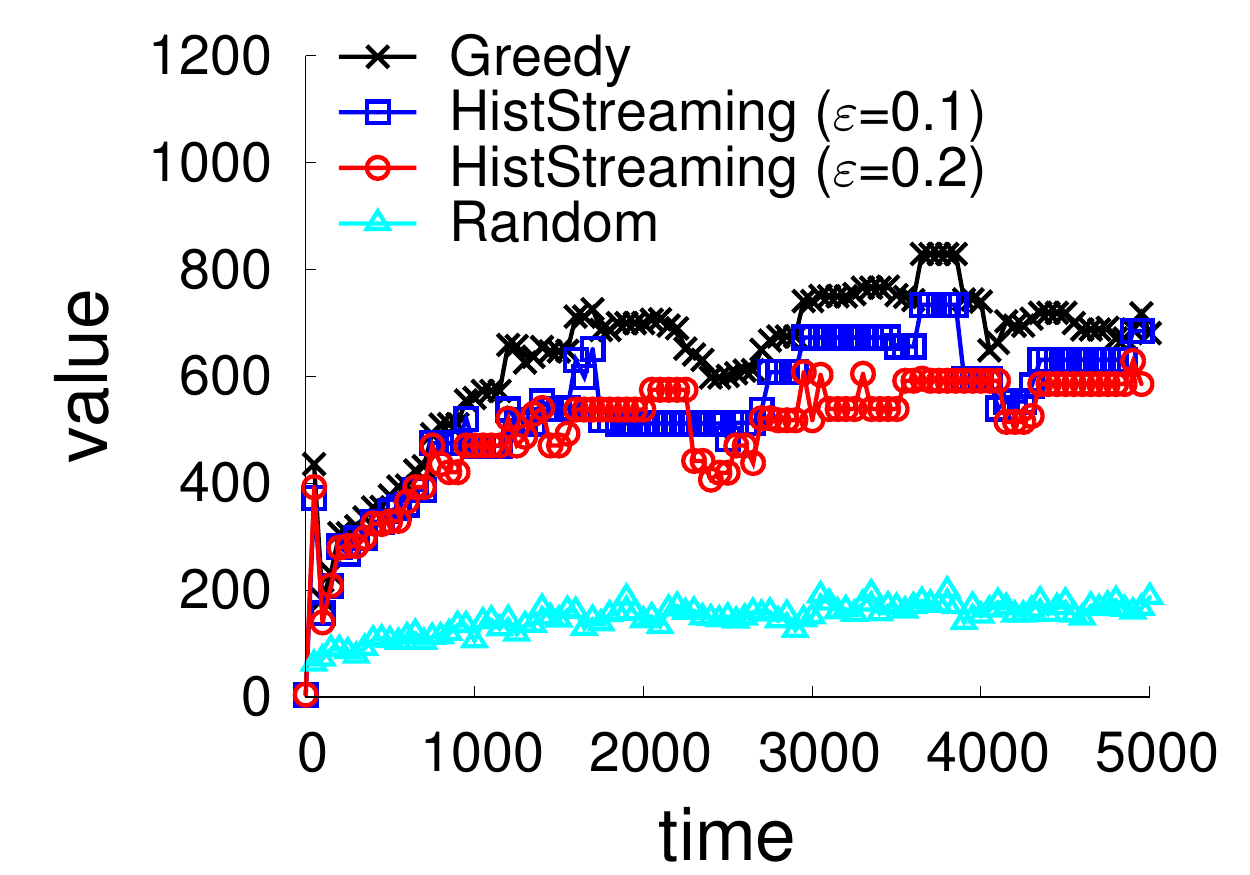}}
\subfloat[StackExchange]{\includegraphics[width=.5\linewidth]{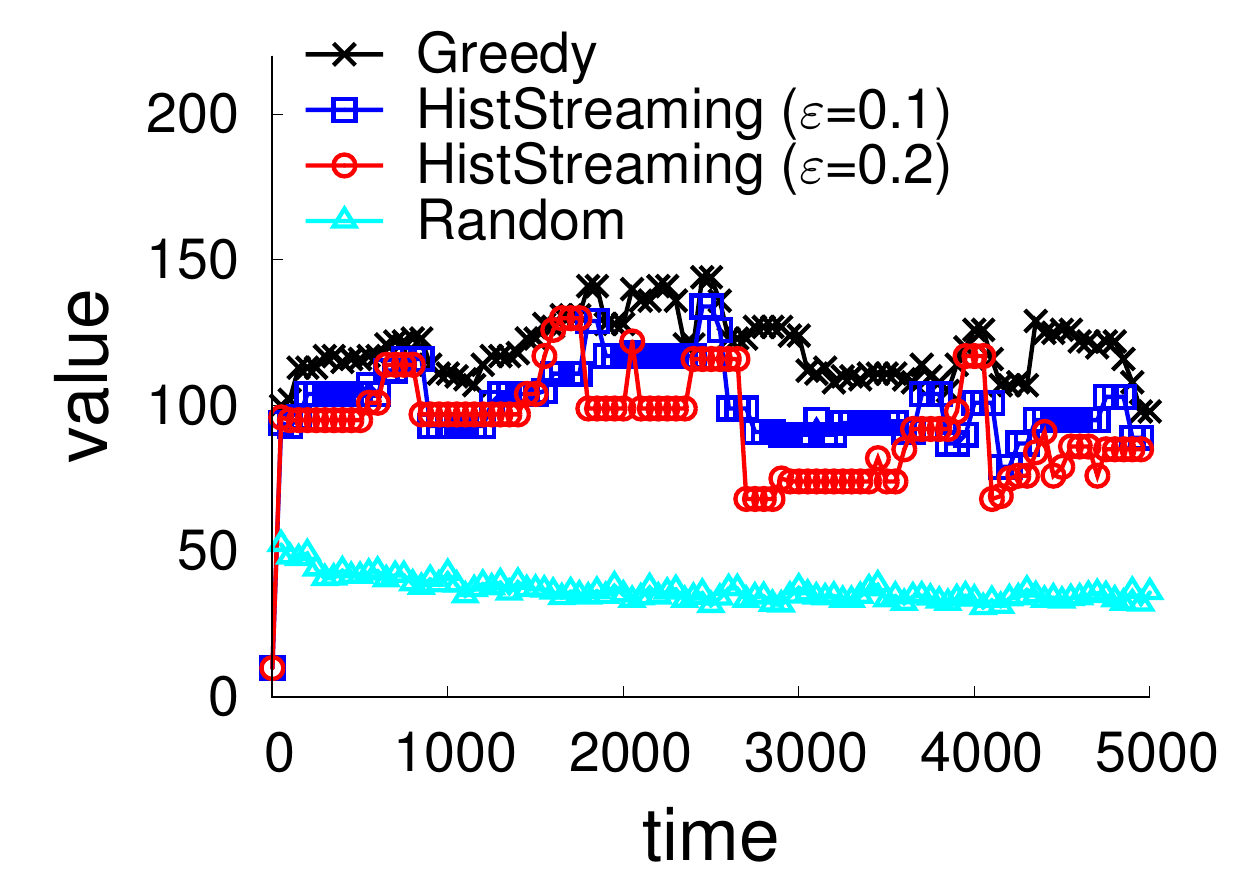}}
\caption{Solution value over time (higher is better)}
\label{fig:val}
\end{figure}

\begin{figure}[htp]
\centering
\subfloat[DBLP]{\includegraphics[width=.5\linewidth]{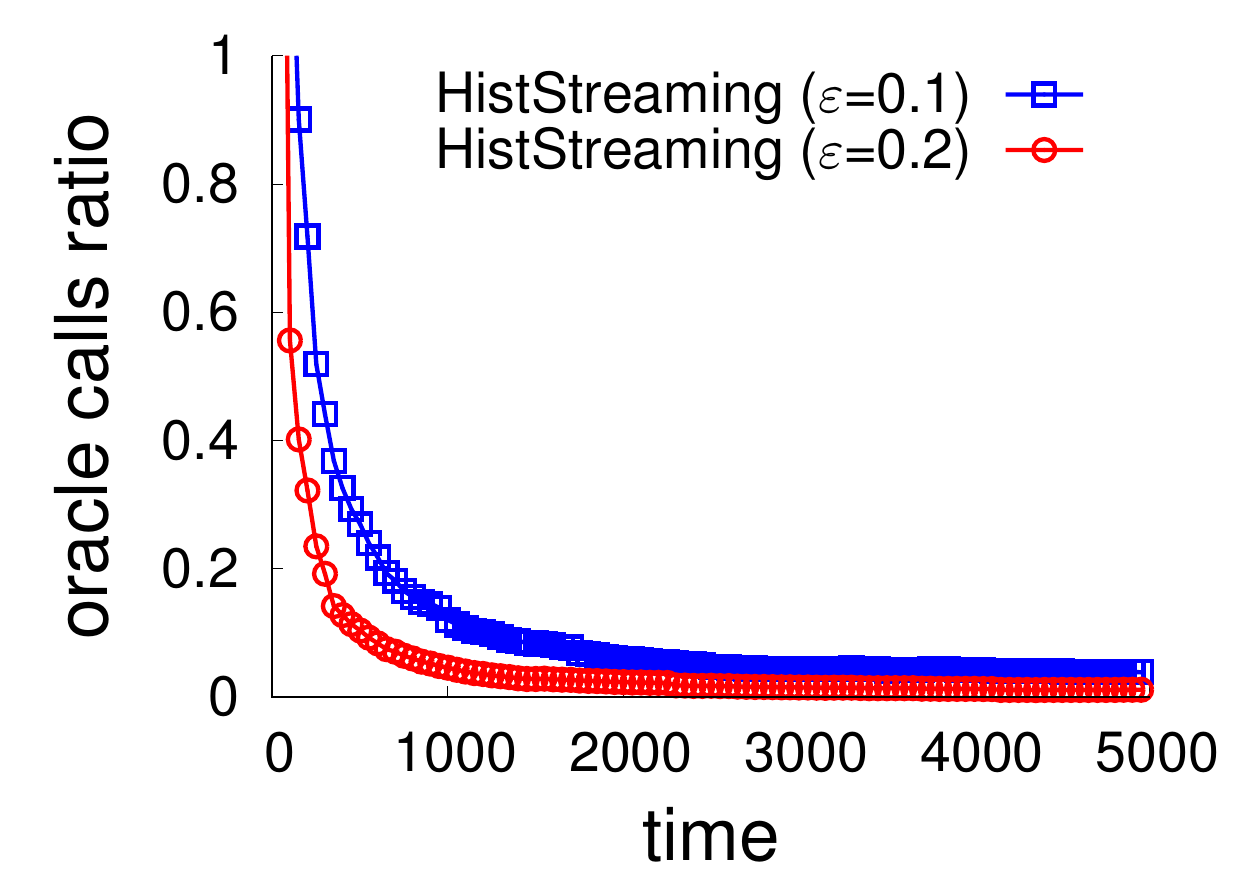}}
\subfloat[StackExchange]{\includegraphics[width=.5\linewidth]{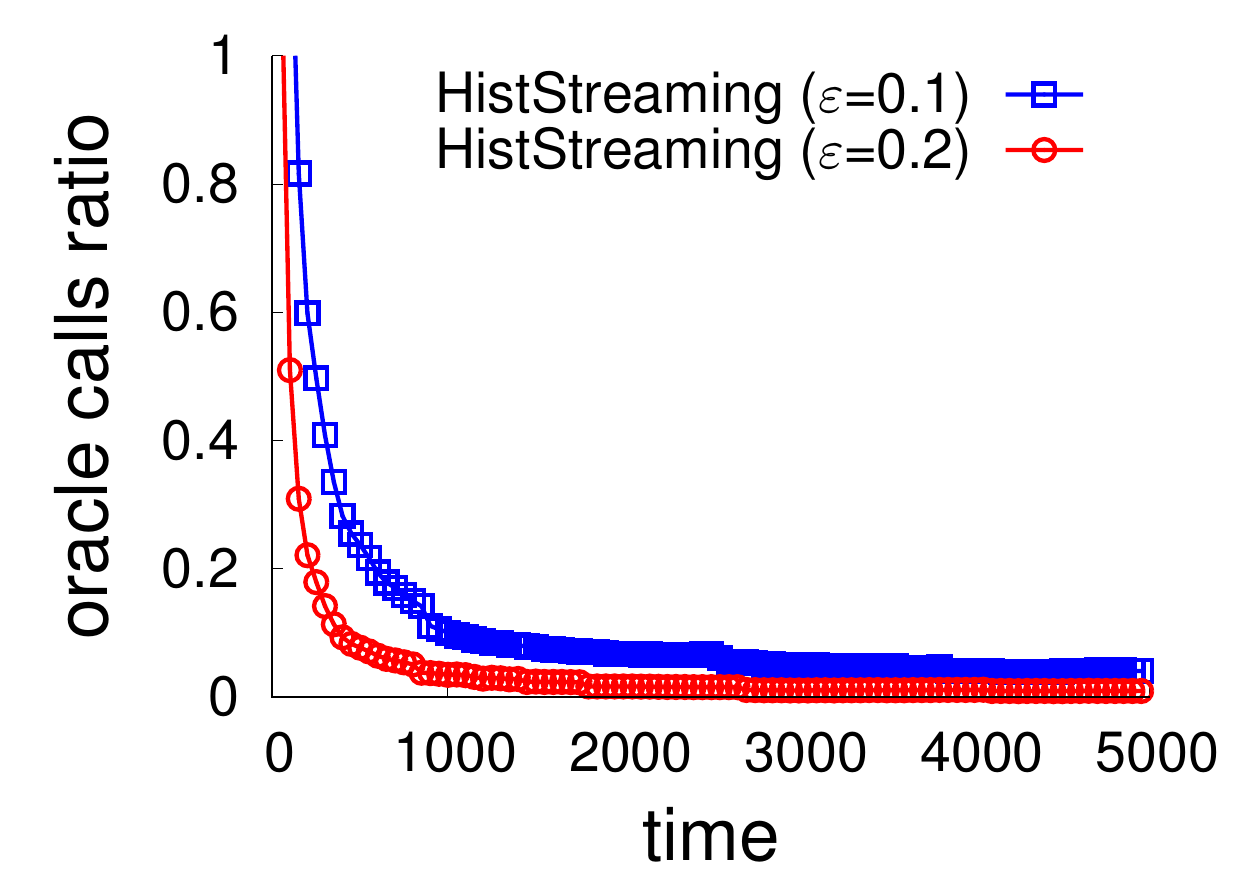}}
\caption{Oracle calls ratio over time (lower is better)}
\label{fig:calls}
\end{figure}

Figure~\ref{fig:val} shows that \textsc{Greedy} and \textsc{Random} always find
the best and worst solutions, respectively, which is expected.
\textsc{HistStreaming} finds solutions that are close to \textsc{Greedy}.
Small $\epsilon$ can further improve the solution quality.
In Fig.~\ref{fig:calls}, we show the ratio of cumulative number of oracle calls
between \textsc{HistStreaming} and \textsc{Greedy}.
It is clear to see that \textsc{HistStreaming} uses quite a small number of oracle
calls comparing with \textsc{Greedy}.
Larger $\epsilon$ further improves efficiency, and for $\epsilon=0.2$ the speedup
of \textsc{HistStreaming} could be up to two orders of magnitude faster than
\textsc{Greedy}.

This experiment demonstrates that \textsc{HistStreaming} finds solutions with
quality close to \textsc{Greedy} and is much more efficient than \textsc{Greedy}.
$\epsilon$ can trade off between solution quality and computational efficiency.

\header{Performance under Different Budget $k$.}
Finally, we conduct experiments to study the performance of \textsc{HistStreaming}
under different budget $k$.
Here, we choose the lifespan distribution as the same as the previous experiment,
and set $\epsilon=0.2$.
We run \textsc{HistStreaming} and \textsc{Greedy} for $1000$ time steps and
compute the ratios of solution value and number of oracle calls between
\textsc{HistStreaming} and \textsc{Greedy}.
The results are depicted in Fig.~\ref{fig:ratio_k}.

\begin{figure}[htp]
\centering
\subfloat[DBLP]{\includegraphics[width=.5\linewidth]{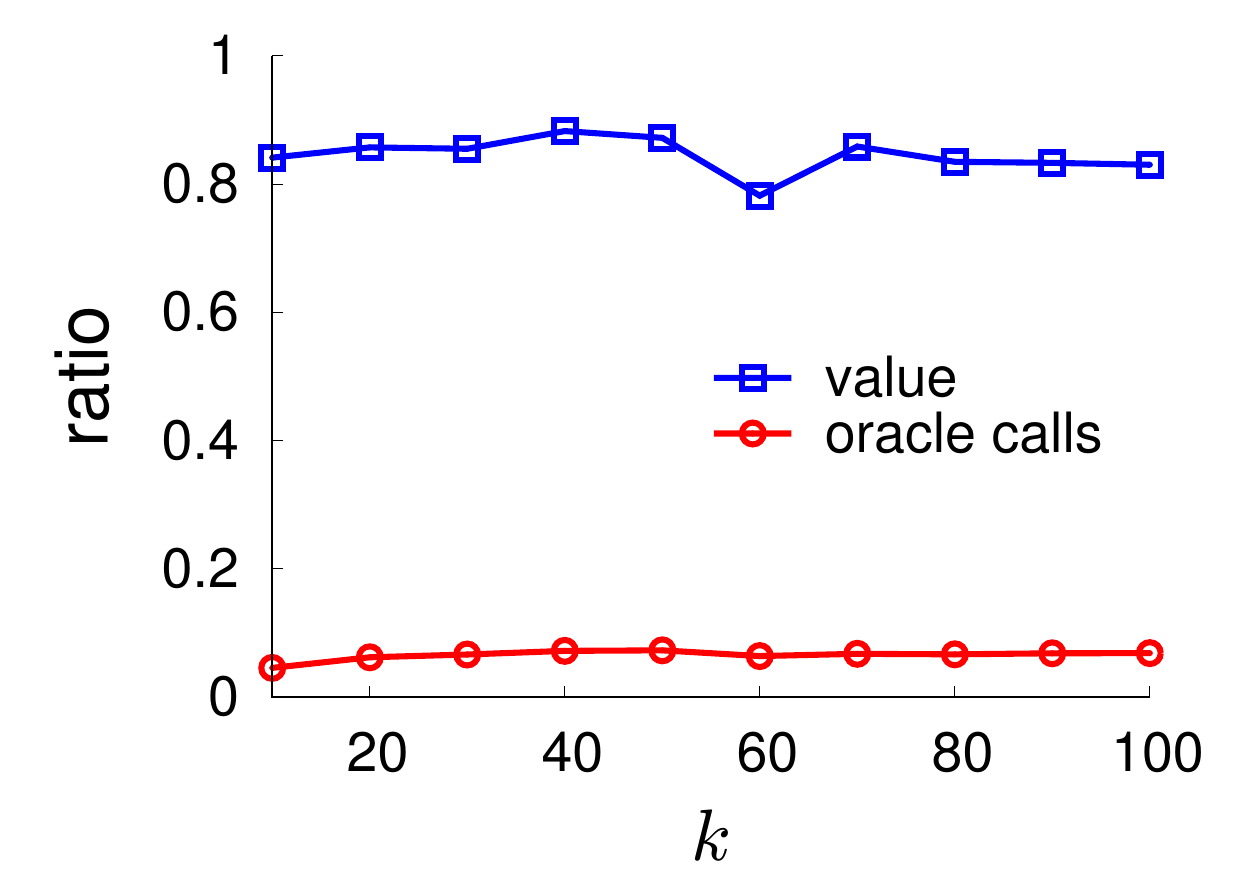}}
\subfloat[StackExchange]{\includegraphics[width=.5\linewidth]{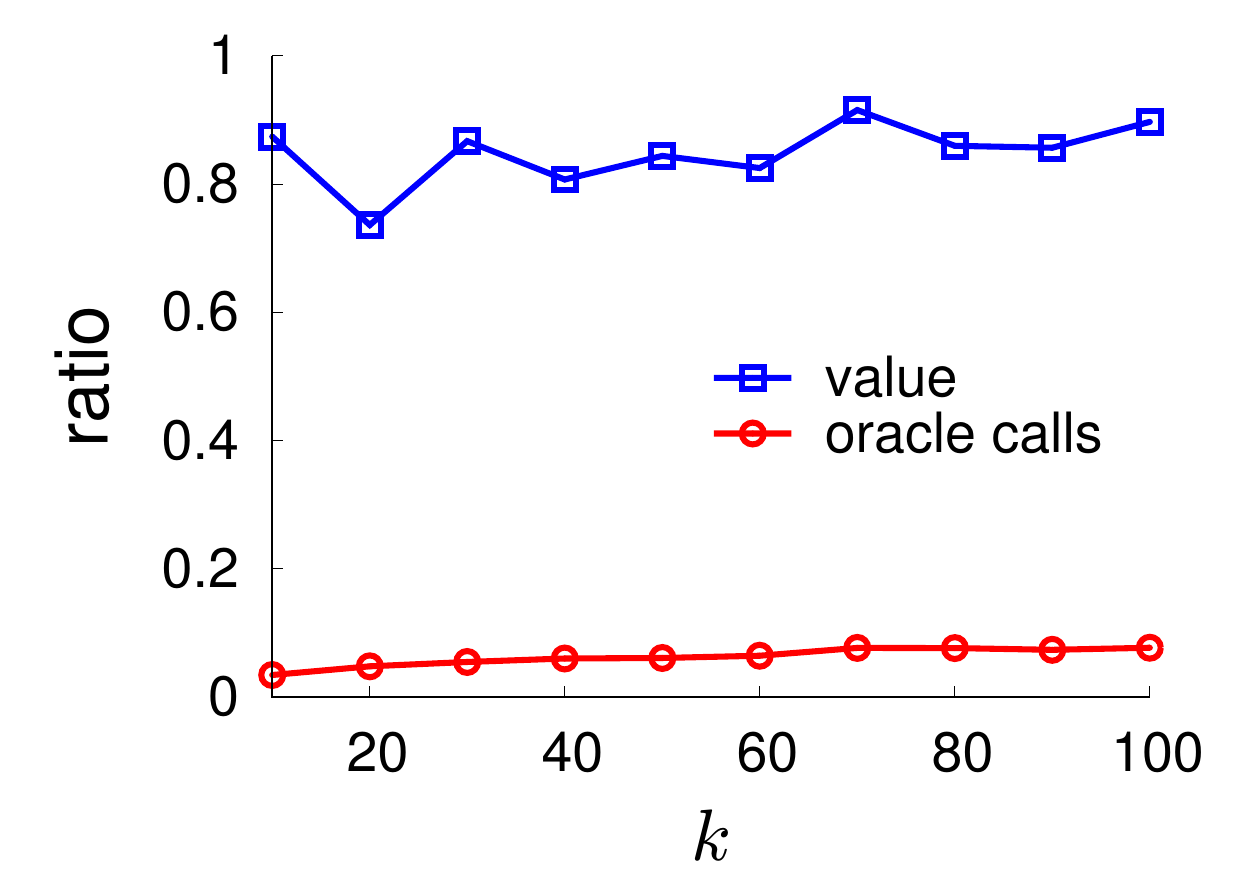}}
\caption{Ratios under different budget $k$}
\label{fig:ratio_k}
\end{figure}

In general, using different budgets, \textsc{HistStreaming} always finds solutions
that are close to \textsc{Greedy}, i.e., larger than $80\%$; but uses very few
oracle calls, i.e., less than $10\%$.
Hence, we conclude that \textsc{HistStreaming} finds solutions with similar
quality to \textsc{Greedy}, but is much efficient than \textsc{Greedy}, under
different budgets.

\section{Related Work}
\label{sec:related_work}

\header{Cardinality Constrained Submodular Function Maximization.}
Submodular optimization lies at the core of many data mining and machine learning
applications.
Because the objectives in many optimization problems have a diminishing returns
property, which can be captured by submodularity.
In the past few years, submodular optimization has been applied to a wide variety
of scenarios, including sensor placement~\cite{Krause2008}, outbreak
detection~\cite{Leskovec2007a}, search result diversification~\cite{Agrawal2009},
feature selection~\cite{Brown2012}, data
summarization~\cite{Mirzasoleiman2015,Mitrovic2018}, influence
maximization~\cite{Kempe2003}, just name a few.
The \textsc{Greedy} algorithm~\cite{Nemhauser1978} plays as a silver bullet in
solving the cardinality constrained submodular maximization problem.
Improving the efficiency of \textsc{Greedy} algorithm has also gained a lot of
interests, such as lazy evaluation~\cite{Minoux1978}, disk-based
optimization~\cite{Cormode2010}, distributed
computation~\cite{Epasto2017b,Kumar2013a}, sampling~\cite{Mirzasoleiman2015}, etc.

\header{Streaming Submodular Optimization (SSO).}
SSO is another way to improve the efficiency of solving submodular optimization
problems, and are gaining interests in recent years due to the rise of big data
and high-speed streams that an algorithm can only access a small fraction of the
data at a time point.
Kumar et al.~\shortcite{Kumar2013a} design streaming algorithms that need to
traverse the streaming data for a few rounds which is suitable for the MapReduce
framework.
Badanidiyuru et al.~\shortcite{Badanidiyuru2014a} then design the
\textsc{SieveStreaming} algorithm which is the first one round streaming algorithm
for insertion-only streams.
\textsc{SieveStreaming} is adopted as the basic building block in our algorithms.
SSO over sliding-window streams has recently been studied by Chen et
al.~\shortcite{Chen2016f} and Epasto et al.~\shortcite{Epasto2017} respectively,
that both leverage smooth histograms~\cite{Braverman2007}.
Our algorithms actually can be viewed as a generalization of these existing
methods, and our SSO techniques apply for streams with inhomogeneous decays.

\header{Streaming Models.}
The sliding-window streaming model is proposed by Datar et
al.~\shortcite{Datar2002}.
Cohen et al.~\shortcite{Cohen2006} later extend the sliding-window model to
general time-decaying model for the purpose of approximating summation aggregates
in data streams (e.g., count the number of $1$'s in a $01$ stream).
Cormode et al.~\shortcite{Cormode2009} consider the similar estimation problem by
designing time-decaying sketches.
These studies have inspired us to propose the IDS model.

\section{Conclusion}
\label{sec:conclusion}

When a data stream consists of elements with different lifespans, existing SSO
techniques become inapplicable.
This work formulates the SSO-ID problem, and presents three new SSO techniques to
address the SSO-ID problem.
\textsc{BasicStreaming} is simple and achieves an $(1/2-\epsilon)$ approximation
factor, but it may be inefficient.
\textsc{HistApprox} improves the efficiency of \textsc{BasicStreaming}
significantly and achieves an $(1/3-\epsilon)$ approximation factor, but it
requires additional memory to store active data.
\textsc{HistStreaming} uses heuristics to further improve the efficiency of
\textsc{HistApprox}, and no longer requires storing active data in memory.
In practice, if memory is not a problem, we suggest using \textsc{HistApprox} as
it has a provable approximation guarantee; otherwise, \textsc{HistStreaming} is
also a good choice.

\section*{Acknowledgment}

We would like to thank the anonymous reviewers for their valuable comments and
suggestions to help us improve this paper.
This work is financially supported by the King Abdullah University of Science and
Technology (KAUST) Sensor Initiative, Saudi Arabia.
The work of John C.S.~Lui was supported in part by the GRF Funding 14208816.

\bibliographystyle{aaai}

\clearpage
\onecolumn

\section*{Proof of Lemma 1}

\begin{proof}
If $x_{i+1}'$ became the successor of $x_i'$ due to the removal of indices between
them at some most recent time $t'\leq t$, then procedure \ReduceRedundancy in
Alg.~2 guarantees that $g_{t'}(x_{i+1}')\geq
(1-\epsilon)g_{t'}(x_i')$ after the removal at time $t'$.
From time $t'$ to $t$, it is also impossible to have data elements with lifespans
between the two indices.
Otherwise we will meet a contradiction: either these elements form redundant {\sc
SieveStreaming} instances again thus $t'$ is not the most recent time as
claimed, or these elements form non-redundant {\sc SieveStreaming} instances thus
$x_i$ and $x_{i+1}$ cannot be consecutive at time $t$.
We thus get \textbf{C2}.

Otherwise $x_{i+1}'$ became the successor of $x_i'$ when one of them is inserted
in the histogram at some time $t'\leq t$.
Without lose of generality, let us assume $x_{i+1}'$ is inserted after $x_i'$ at
time $t'$.
If elements with lifespans between the two indices arrive from time $t'$ to $t$,
these elements must form redundant {\sc SieveStreaming} instances.
We still get \textbf{C2}.
Or, there is no element with lifespan between the two indices at all, i.e.,
$\CS_t$ contains no data with lifespan between $x_i$ and $x_{i+1}$.
We thus get \textbf{C1}.
This completes the proof.
\end{proof}

\section*{Proof of Theorem~3}

\begin{proof}
If $x_1=1$ at time $t$, then $\CA_1^{(t)}$ exists.
By the property of \textsc{SieveStreaming}, we conclude that
\[
g_t(x_1)=g_t(1)\geq (\frac{1}{2} - \epsilon)\OPT_t.
\]

Otherwise we have $x_1>1$ at time $t$.
If $\CS_t$ contains elements with lifespan less than $x_1$, then $\CA_{x_1}^{(t)}$
does not process all of the elements in $\CS_t$, thus incurs a loss of solution
quality.
Our goal is to bound this loss.

Let $x_0$ denote the most recent expired predecessor of $x_1$ at time $t$, and let
$t'<t$ denote the last time $x_0$'s ancestor $x_0'$ was still alive, i.e.,
$x_0'=1$ at time $t'$ (cf.~Fig.~\ref{fig:proof_histapprox}).
For ease of presentation, we commonly refer to $x_0$ and $x_0$'s ancestors as the
left index, and refer to $x_1$ and $x_1$'s ancestors as the right index.
Obviously, in time interval $(t',t]$, no element with lifespan less than the right
index arrives; otherwise, these elements would create new indices before the right
index; then $x_1$ will not be the first index at time $t$, or $x_0$ is not the
most recent expired predecessor of $x_1$ at time $t$.

\begin{figure}[htp]
\centering
\begin{tikzpicture}[
scale=.8,
txt/.style={font=\footnotesize, inner sep=2pt, text depth=.25ex, align=center},
ntxt/.style={txt,anchor=north},
arr/.style={densely dashed, -stealth},
]

\node[txt] (tpp) at (-1.5,2) {$t''$};
\draw[thick] (0,2) -- (5,2);
\draw[thick] (0,2.1) -- ++(0,-.1) node[ntxt]{$1$};
\draw[thick, purple] (.5,2.1) -- ++(0,-.1) node[ntxt,purple] (x0pp) {$x_0''$};
\draw[thick, blue] (1.5,2.1) -- ++(0,-.1) node[ntxt,blue] (x1pp) {$x_1''$};

\node[txt] (tp) at (-1.5,1) {$t'$};
\draw[thick] (0,1) -- (5,1);
\draw[thick, purple] (0,1.1) -- ++(0,-.1)
node[ntxt,purple] (x0p) {$\,\,\,x_0'\!=\!1$};
\draw[thick, blue] (1,1.1) -- ++(0,-.1) node[ntxt,blue] (x1p) {$x_1'$};

\node[txt] (t) at (-1.5,0) {$t$};
\draw[thick] (0,0) -- (5,0);
\draw[thick,dotted, gray] (-.5,0) -- (0,0);
\draw[thick, purple!60] (-.5,.1) -- ++(0,-.1) node[ntxt,purple!60] {$x_0$};
\draw[thick] (0,.1) -- ++(0,-.1) node[ntxt] {$1$};
\draw[thick, blue] (.5,.1) -- ++(0,-.1) node[ntxt,blue] (x1) {$x_1$};

\draw[arr, blue] (x1pp) -- (1.05,1.1);
\draw[arr, blue] (x1p) -- (.55,.1);
\draw[arr, purple] (x0pp) -- (.05,1.1);
\draw[arr, purple!60] (x0p) -- (-.45,.1);

\draw[->] (tpp) -- (tp);
\draw[->] (tp) -- (t);

\node[txt,purple,font=\scriptsize] at (-.5,1.5) {left index};
\node[txt,blue,font=\scriptsize] at (2,.8) {right index};
\end{tikzpicture}

\caption{Indices relations at time $t''\leq t'<t$.
For ease of presentation, $x_0$ (resp.~$x_1$) and its ancestors will be
simply referred to as the left (right) index.
}
\label{fig:proof_histapprox}
\end{figure}
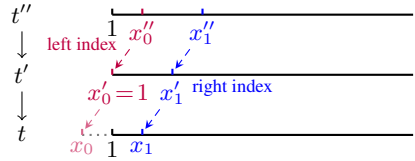

Notice that $x_0'$ and $x_1'$ are two consecutive indices at time $t'$.
By Lemma~1, we have two cases.

\bullethdr{If C1 holds.}
In this case, $\CS_{t'}$ contains no element with lifespan between $x_0'$ and
$x_1'$.
Because there is also no element with lifespan less than the right index from time
$t'$ to $t$, then $\CS_t$ has no element with lifespan less than $x_1$ at time
$t$.
Therefore, $\CA_{x_1}^{(t)}$ processed all of the elements in $\CS_t$.
By the property of \textsc{SieveStreaming}, we still have
\[
g_t(x_1) = g_t(1)\geq (\frac{1}{2} - \epsilon)\OPT_t.
\]

\bullethdr{If C2 holds.}
In this case, there exists some time $t''\leq t'$ s.t.
$g_{t''}(x_1'')\geq (1 - \epsilon)g_{t''}(x_0'')$ holds
(cf.~Fig.~\ref{fig:proof_histapprox}), and from time $t''$ to $t'$, no element
with lifespan between the two indices arrived (however $\CS_t$ may have elements
with lifespan less than $x_1$ at time $t$ and these elements arrived before time
$t''$).
Notice that elements with lifespan no larger than the left index all expired after
time $t'$ and they do not affect the solution at time $t$; therefore, we can
safely ignore these elements in our analysis and only care elements with lifespan
no less than the right index arrived in interval $[t'',t]$.
Notice that these elements are only inserted on the right side of the right index.

In other words, at time $t''$, the output values of the two instances satisfy
$g_{t''}(x_1'')\geq (1 - \epsilon)g_{t''}(x_0'')$; from time $t''$ to $t$, the two
instances are fed with same elements.
Such a scenario has been studied in the sliding-window case~\cite{Epasto2017}.
By submodularity of $f$ and suffix-monotonicity\footnote{That is, feeding more
elements to a \textsc{SieveStreaming} algorithm cannot decrease its output
value.}
of the \textsc{SieveStreaming} algorithm, the following lemma guarantees that
$g_t(x_1)$ is close to $\OPT_t$.

\begin{lemma}\label{lem:smooth_histogram}
Consider a cardinality constrained monotone submodular function maximization
problem.
Let $\CA(S)$ denote the output value of applying the \textsc{SieveStreaming}
algorithm on stream $S$.
Let $S\Vert S'$ denote the concatenation of two streams $S$ and $S'$.
If $\CA(S_2)\geq(1-\epsilon)\CA(S_1)$ for $S_2\subseteq S_1$ (i.e., each element
in stream $S_2$ is also an element in stream $S_1$), then $\CA(S_2\Vert S)\geq
(1/3-\epsilon)\OPT$ for all $S$, where $\OPT$ is the value of an optimal
solution in stream $S_1\Vert S$.
\end{lemma}

In our scenario, at time $t''$ the two instances satisfy $g_{t''}(x_1'')\geq (1 -
\epsilon)g_{t''}(x_0'')$ and $\CA_{x_1''}$'s input elements is a subset of
$\CA_{x_0''}$'s input elements.
After time $t''$, the two instances are fed with same elements.
Hence, $g_t(x_1)\geq (1/3-\epsilon)\OPT_t$.

Combining above results, we conclude that {\sc HistApprox} guarantees a $(1/3 -
\epsilon)$ approximation factor.
\end{proof}

\section*{Proof of Theorem~4}

\begin{proof}
At any time $t$, because $g_t(x_{i+2})<(1-\epsilon)g_t(x_i)$, and $g_t(l)\in
[\Delta, k\Delta]$ where $\Delta\triangleq\max\{f(v)\colon v\in V\}$, then the
size of index set $\bx_t$ is upper bounded by
$O(\log_{(1-\epsilon)^{-1}}k)=O(\epsilon^{-1}\log k)$.
For each data element, in the worst case, we need to update $|\bx_t|$ {\sc
SieveStreaming} instances, and each {\sc SieveStreaming} instance has update
time $O(\epsilon^{-1}\log k)$.
In addition, procedure \ReduceRedundancy has complexity is
$O(\epsilon^{-2}\log^2 k)$.
Thus the total time complexity is $O(\epsilon^{-2}\log^2 k)$.

For memory usage, because {\sc HistApprox} maintains $|\bx_t|$ {\sc
SieveStreaming} instances, and each instance uses memory $O(k\epsilon^{-1}\log
k))$.
Thus the total memory used by {\sc HistApprox} is $O(k\epsilon^{-2}\log^2k))$.
\end{proof}

\section*{Proof of Lemma 2}

We state a more general conclusion than Lemma 2.
This conclusion will be useful in the later discussion.

\begin{lemma}\label{lem:general_substreams}
Let $S_1\Vert S_2\Vert S_3$ denote the concatenation of three streams
$S_1,S_2,S_3$.
Let $\CA$ denote an insertion-only SSO algorithm with approximation factor $c$.
Let $\CA(S)$ denote the output value of applying algorithm $\CA$ on stream $S$.
Assume $\CA$ is suffix-monotone, i.e., $\CA(S\Vert S')\geq\CA(S), \forall S'$.
Let $\OPT$ denote the value of an optimal solution in stream $S_1\Vert S_2\Vert
S_3$.
Then we have the following results:
\begin{description}
\item[(1)] If $\CA(S_2)\geq(1-\epsilon)\CA(S_1\Vert S_2)$, then $\CA(S_2\Vert
S_3)\geq\frac{c}{2}(1-\epsilon)\OPT$.\footnote{Adapted from Lemma~1
in~\cite{Chen2016f}.}
\item[(2)] If $\CA(S_1)\geq(1-\epsilon)\CA(S_1\Vert S_2)$, then $\CA(S_1\Vert
S_3)\geq\frac{c}{2}(1-\epsilon)\OPT$.
\end{description}
\end{lemma}
\begin{proof}
Let $O_{123},O_{12},O_{13},O_{23}$, and $O_3$ denote the optimal solutions in
streams $S_1\Vert S_2\Vert S_3$, $S_1\Vert S_2$, $S_1\Vert S_3$, $S_2\Vert S_3$,
and $S_3$, respectively.

To prove (1), by the property of algorithm $\CA$, we have
\[
\CA(S_2\Vert S_3)\geq cf(O_{23}).
\]
We also have
\begin{align*}
\CA(S_2\Vert S_3)
&\geq \CA(S_2) \\
&\geq (1-\epsilon)\CA(S_1\Vert S_2) \\
&\geq c(1-\epsilon)f(O_{12}).
\end{align*}
Combining above two relations, we have
\begin{align*}
2\CA(S_2\Vert S_3)
&\geq c(1-\epsilon)f(O_{12}) + cf(O_{23}) \\
&\geq c(1-\epsilon)[f(O_{12}) + f(O_{23})] \\
&\geq c(1-\epsilon)[f(O_{12}) + f(O_3)] \\
&\geq c(1-\epsilon)[f(O_{123}\cap O_{12}) + f(O_{123}\cap O_3)] \\
&\geq c(1-\epsilon)f(O_{123}) \\
&= c(1-\epsilon)\OPT
\end{align*}
where the last inequality holds due to the submodularity of $f$.\footnote{For two
sets $A,B\subseteq V$ and a submodular function $f\colon
2^V\mapsto\mathbb{R}$, it holds that $f(A)+f(B)\geq f(A\cup B)+f(A\cap B)$.}
We hence obtain
\[
\CA(S_2\Vert S_3)\geq\frac{c}{2}(1-\epsilon)f(O_{123}).
\]

To prove (2), similarly, by the property of algorithm $\CA$, we have
\[
\CA(S_1\Vert S_3)\geq cf(O_{13}).
\]
We also have
\begin{align*}
\CA(S_1\Vert S_3)
&\geq \CA(S_1) \\
&\geq (1-\epsilon)\CA(S_1\Vert S_2) \\
&\geq c(1-\epsilon)f(O_{12}).
\end{align*}
Combining above two relations, we have
\begin{align*}
2\CA(S_1\Vert S_3)
&\geq c(1-\epsilon)f(O_{12}) + cf(O_{13}) \\
&\geq c(1-\epsilon)[f(O_{12}) + f(O_{13})] \\
&\geq c(1-\epsilon)[f(O_{12}) + f(O_3)] \\
&\geq c(1-\epsilon)[f(O_{123}\cap O_{12}) + f(O_{123}\cap O_3)] \\
&\geq c(1-\epsilon)f(O_{123}) \\
&= c(1-\epsilon)\OPT.
\end{align*}
We hence obtain
\[
\CA(S_1\Vert S_3)\geq\frac{c}{2}(1-\epsilon)f(O_{123}).
\]
\end{proof}

If $\CA$ is the \textsc{SieveStreaming} algorithm, then $c=1/2-\epsilon$.
In (2), replace the condition $\CA(S_1)\geq(1-\epsilon)\CA(S_1\Vert S_2)$ by
$\CA(S_1)\geq\alpha\CA(S_1\Vert S_2)$, then we immediately obtain Lemma~2.

Conclusion (1) of Lemma~\ref{lem:general_substreams} is a generalization of
Lemma~\ref{lem:smooth_histogram}.
Lemma~\ref{lem:smooth_histogram} is specific to \textsc{SieveStreaming} and the
proof of Lemma~\ref{lem:smooth_histogram} leverages a specific property of the
\textsc{SieveStreaming} algorithm.
Hence, a better bound is obtained.

Note that conclusions (1) and (2) in Lemma~\ref{lem:general_substreams} can also
be unified.

\begin{lemma}\label{lem:general_substreams2}
Let $S\Vert S'$ denote the concatenation of two streams $S$ and $S'$.
Let $\CA$ denote an insertion-only SSO algorithm with approximation factor $c$.
Let $\CA(S)$ denote the output value of applying algorithm $\CA$ on stream $S$.
Assume $\CA$ is suffix-monotone, i.e., $\CA(S\Vert S')\geq\CA(S), \forall S'$.
If two streams $S_1,S_2$ have relation $S_2\subseteq S_1$, i.e., each element in
stream $S_2$ is also an element in stream $S_1$, and
$\CA(S_2)\geq(1-\epsilon)\CA(S_1)$, then $\CA(S_2\Vert
S)\geq\frac{c}{2}(1-\epsilon)\OPT$ for all stream $S$, where $\OPT$ is the value
of an optimal solution in stream $S_1\Vert S$.
\end{lemma}

\section*{A Note On {\sc HistApprox} and {\sc HistStreaming}}

\textsc{HistApprox} maintains a histogram satisfying Lemma~1, which ensures two
consecutive indices $x_0$ and $x_1$ to satisfy either Case 1 or Case 2.

Case 1 is trivial.
It simply states that there is no data with lifespan between $x_0$ and $x_1$.
This property is used in the proof to show that $\CA_{x_1}^{(t)}$ processes all
the elements in $\CS_t$.
Hence, $g_t(x_1)\geq (1/2-\epsilon)\OPT_t$.

Case 2 states that, $g_{t'}(x_1')\geq (1-\epsilon)g_{t'}(x_0')$ at some time $t'$,
and from $t'$ to $t$, $\CA_{x_0}$ and $\CA_{x_1}$ are fed with same elements.
Lemma~\ref{lem:smooth_histogram} then guarantees that their output values will
remain close with each other after time $t'$, and $g_t(x_1)\geq
(1/3-\epsilon)\OPT_t\geq(1/3-\epsilon)g_t(x_0)$.

As we discussed in the proof of Lemma~\ref{lem:general_substreams}, conclusion (1)
of Lemma~\ref{lem:general_substreams} generalizes
Lemma~\ref{lem:smooth_histogram}.
Hence, if $\CA_{x_0}$ and $\CA_{x_1}$ in \textsc{HistApprox} could find solutions
with constant approximation factors, then in Case 2, their output values will be
still close with each other, and $g_t(x_1)$ has a constant approximation factor to
$\OPT_t$.

\textsc{HistStreaming} algorithm essentially leverages above observation.
\textsc{HistStreaming} ignores the insignificant historical data and conclusion
(2) of Lemma~\ref{lem:general_substreams} guarantees that each instance in
\textsc{HistStreaming} still has a constant approximation factor.
Then when we were in Case 2, conclusion (1) of Lemma~\ref{lem:general_substreams}
will guarantee that $g_t(x_1)$ has a constant approximation factor to $\OPT_t$.

\end{document}